\documentclass{elsart}

\usepackage{amssymb}
\usepackage[centertags]{amsmath}

\usepackage[pdftex]{graphicx}
\usepackage{colortbl}
\makeatother

\newtheorem{theorem}{Theorem}[section]
\newtheorem{lemma}[theorem]{Lemma}

\newtheorem{proposition}[theorem]{Proposition}
\newtheorem{remark}[theorem]{Remark}

\newtheorem{definition}[theorem]{Definition}

\newenvironment{proof}{\removelastskip\par\medskip}
{\penalty-20\null\hfill$\square$\par\medbreak}

\begin{document}

\begin{frontmatter}



\title{ Minimum penalized Hellinger distance for model selection in small samples}
\thanks[talk]{ This research was supported, in part, by grants from 
 {\bf AIMS}(African Institute for Mathematical Sciences) 6 Melrose Road, Muizenberg-Cape Town 7945 South Africa }

\author{\small Papa Ngom$^{a}$,  Bertrand Ntep$^{b}$}

\address{$^{a,}$ $^{b}${ \footnotesize 
 LMA -
 Laboratoire de Math\'ematiques Appliqu\'ees \\
 Universit\'e Cheikh Anta Diop \\ BP 5005 Dakar-Fann S\'en\'egal}\\
 $^{a}$ e-mail : papa.ngom@ucad.edu.sn \\
 $^{b}$	ntepjojo@yahoo.fr }

\begin{abstract}
\par 
 In statistical modeling area,  the Akaike information criterion AIC, is a widely known and extensively used tool for model choice. The $\phi$-divergence test statistic is a recently developed tool for statistical model selection. The popularity of the divergence criterion is however tempered by their known lack of robustness in small sample. In this paper the penalized minimum Hellinger distance type statistics are considered and some properties are established.  The limit laws of the estimates and test statistics are given under both the null and the alternative hypotheses, and approximations of the power functions are deduced. A model selection criterion relative to these  divergence measures are   developed for parametric inference. Our interest is in the problem to testing for choosing between two models using some informational type statistics, when independent sample are drawn from a discrete population. Here, we discuss the asymptotic properties and the performance of new procedure tests and investigate their small sample behavior. 
 \\  
\end{abstract}

\begin{keyword}
Generalized information, estimation, hypothesis test, Monte Carlo simulation.\\
  AMS Subject Classification : 62F03, 62F05,  60F40,94A17.
\end{keyword}
\end{frontmatter}
\section{Introduction }
 \par 
 A comprehensive surveys on Pearson Chi-square type statistics 
 has been provided by many authors as Cochran (1952), Watson (1956) 
 and Moore (1978,1986), in particular on quadratics forms in the cell frequencies.  
 Recently, Andrews(1988a, 1988b) has extended the Pearson chi-square testing method to non-dynamic
parametric models, i.e., to models with covariates. Because Pearson chi-square statistics
provide natural measures for the discrepancy between the observed data and a specific
parametric model, they have also been used for discriminating among competing models. 
Such a situation is frequent in Social Sciences where many competing models are
proposed to fit a given sample. A well know difficulty is that each chi-square statistic 
tends to become large without an increase in its degrees of freedom as the sample
size increases. As a consequence goodness-of-fit tests based on Pearson type chi-square
statistics will generally reject the correct specification of every competing model.
\par To circumvent such a difficulty, a popular method for model selection, which is
similar to use of Akaike (1973) Information Criterion (AIC), consists in considering that
the lower the chi-square statistic, the better is the model.
The preceding selection rule, however, does not take into account random variations inherent in the values of the statistics.  

We propose here a procedure for taking into account the stochastic nature of these differences so as to assess their significance. The main propose of this paper is to address this issue. We shall propose some convenient asymptotically standard
normal tests for model selection based on $\phi-$divergence type statistics. Following Vuong (1989, 1993), the procedures considered here are testing the null hypothesis that the competing models are equally close to the  data generating process (DGP) versus the alternative hypothesis that one model is closer to the DGP where closeness of a model is measured according to the discrepancy implicit in the $\phi-$divergence type statistic used. Thus the outcomes of our tests provide information on the strength of the statistical evidence for the choice of a model based on its goodness-of-fit.
The model selection approach proposed here differs from those of Cox (1961, 1962) and Akaike (1974) for non nested hypotheses. 
This difference is that the present approach is based on the discrepancy implicit in the divergence type statistics used, 
while these other approaches as Vuong's (1989) tests for model selection rely on the Kullback-Leibler (1951) information criterion (KLIC). \\
Beran (1977) showed that by using the minimum Hellinger distance estimator, one can simultaneously obtain asymptotic efficiency and robustness properties in the presence of outliers. The works of Simpson (1989) and Lindsay (1994) have shown that, in the tests hypotheses, robust alternatives to the likelihood ratio test can be generated by using the Hellinger distance. 
 We consider a general class of estimators that is very broad and contains most of estimators currently used in practice when forming divergence type statistics. This covers the case studies in Harris and Basu (1994); Basu et al. (1996); Basu and Basu  (1998) where the penalized Hellinger distance is used.\\ 
 The remainder of this paper is organized as follows. Section 2 introduces the basic notations and definitions. Section 3
 gives a short overview of divergence measures. Section 4 investigates the asymptotic distribution of the penalized Hellinger distance. In section 5, some applications for testing hypotheses are proposed. Section 6 presents some simulation results. Section 7 concludes the paper.
\section{Definitions and notation}
In this section, we briefly present the basic assumptions on the model and parameters estimators, and we define our generalized divergence type statistics. \\
We consider a discrete statistical model, i.e $X_1,X_2,\ldots X_n$ an independent random sample from a discrete population with support  $ \mathcal{X} = \{ 1,\ldots,m \}$. Let  $  P = \left( p_1,\ldots, p_m\right)^T$ 
be a probability vector i.e $P\in \Omega_m$ where $\Omega_m$  is the simplex of probability m-vectors, $$\Omega_m=\big\{ \left( p_1,p_2,\ldots,p_m \right)\in \mathbb{R}^ m \ ;  \displaystyle \sum_{i=1}^{m}p_i=1,\ p_i\geq 0, i=1,\dots,m \big\}.$$
We consider a parameter model 
$$ \mathcal{P}= \{ P_{\theta} = \left( p_1(\theta),\ldots, p_m(\theta)\right)^T: \ \theta \in \Theta   \}$$
which may or may not contain the true distribution $P$, where $\Theta$ is a compact subset of k-dimensional Euclidean space (with $k <m-1$). If $ \mathcal{P}$ cointains $P$, then there exists a $\theta_0 \in \Theta$ such that $P_{\theta_0} = P$ and the model $\mathcal{P}$ is said to be correctly specified.

We are interested in testing $$H_0 : P \in {\cal{P}}\  (\hbox{ with true parameter} \ {\theta_0} ) \ \hbox{ versus }  H_1 : P \in \Omega_m - \cal{P}.$$
By $\parallel \cdot \parallel$ we denote the usual Euclidean norm and we interpret probability distributions on $\mathcal{X}$ as row vectors from $\mathbb{R}^m$. For simplicity we restrict ourselves to unknown true parameters $\theta_0$ satisfying the classical regularity conditions given by Birch (1964):
\par 1. True $\theta_0$ is an interior point of $\Theta$ and $ p_{i \theta_0} >0$ for $i=1,\ldots,m$. 
Thus $P_{\theta_0}=\left( p_{1 \theta_0},\ldots,p_{m_{\theta_0}} \right)^T$ is an interior point of the set $\Omega_m$.
\par 2. The mapping  $P:\Theta\longrightarrow\Omega_m$ is totally differentiable at $\theta_0$ so that the partial derivatives of $p_i$ with respect to each $\theta_j$ exist at $\theta_0$ and $p_i(\theta)$ has a linear approximation at $\theta_0$ given by 
$$   p_i(\theta)= p_i(\theta_0)+ \sum_{j=1}^{k} (\theta_j-\theta_{0j})\frac{\partial p_i(\theta_0)}{\partial \theta_j}+o(\parallel\theta-\theta_{0}\parallel)$$
where $ o(\parallel\theta-\theta_{0}\parallel) \hbox{ denotes a function verifying }
 \displaystyle \lim_{\theta\longrightarrow\theta_{0}} \frac{o(\parallel\theta-\theta_{0}\parallel)}{\parallel\theta-\theta_{0}\parallel}=0.$
\par 3. The Jacobian matrix 
$  \displaystyle J(\theta_0)= \left(\dfrac{\partial P_{\theta}}{\partial \theta} \right)_{\theta=\theta_0} = \left( \frac{\partial p_i(\theta_0)}{\partial \theta_j}\right)_{\substack{1\leq i \leq m \\  1\leq j \leq k} } $
is of full rank (i.e.  of rank k and $k < m$).
\par 4. The inverse mapping $P^{-1}:{\cal{P}} \longrightarrow\Theta$ is continuous at $P_{\theta_0}.$
\par 5. The mapping $P:\Theta\longrightarrow\Omega_m$ is continuous at every point $\theta \in \Theta$.\\  
 
 Under the hypothesis that $P\in \mathcal{P}$, there exists an unknown parameter $\theta_0$ such that $P=P_{\theta_0}$ and the problem of point estimation appears in a natural way. Let $n$ be sample size. We can estimate the distribution  
$ P_{\theta_0}=\left( p_1(\theta),p_2(\theta),\ldots,p_m(\theta) \right)^T $
by  the vector of  observed frequencies  $\widehat{P}=(\hat{p}_1,\ldots,\hat{p}_m)$ on $\mathcal{X}$ ie of measurable mapping $\mathcal{X}^n\longrightarrow \Omega_m$.  \\
This non parametric estimator $\widehat{P}=(\hat{p}_1,\ldots,\hat{p}_m)$ is defined by 
\begin{equation}
\hat{p}_j=\frac{N_j}{n}, \quad N_j=\displaystyle {\sum_{i=1}^n} T^i_j(X_i)\label{eq1} 
 \hbox{ \ where \   } T^i_j(X_i) =\left\lbrace\begin{array}{lll}
1 & & \hbox{if } X_i =j\\
0& &  \hbox{otherwise}
\end{array}\right.
\end{equation}
 We can now define the class of $ \phi$-divergence type statistics considered in this paper.
\section{A brief review of $\phi$-divergences }
Many different measures quantifying the degree of discrimination between two probability distributions have been studied in the past. They are frequently called distance measures, although some of them are not strictly metrics. They have been applied to different areas, such as medical image registration (Josien P.W. Pluim, 2001), classification and retrieval, among others. This class of distances is referred, in the literature, as the class of $ \phi$, f or g-divergences (Csisz$ \grave{a}$r (1967); Vajda  (1989); Morales et al. (1995); Pardo (2006); Bassetti et al. (2007)) or the class of disparities (Lindsay (1994)). The divergence measures play an important role in statistical theory, specially in large theories of estimation and testing. \par Later many papers have appeared in the literature, where divergence or entropy type measures of information have been used in testing statistical hypotheses.
Among others we refer to McCulloch (1988), Read and Cressie (1988), Zografos et al. (1990), Salicr$\grave{u}$ et al. (1994), Bar-Hen and Daudin (1995), Men$\grave{e}$ndez et al. (1995, 1996, 1997), Pardo et al. (1995), Morales et al. (1997, 1998), Zografos (1994, 1998), Bar-Hen
(1996) and the references therein.
A measure of discrimination between two probability distributions called $\phi$-divergence, was introduced by Csisz$\acute{a}$r  (1967). 
\par Recently, Broniatowski et al. (2009) presented a new dual representation for divergences. Their aim was to introduce estimation and test procedures through divergence optimization for discrete or continuous parametric models. In the problem where independent samples are drawn from two different discrete populations, Basu et al. (2010) developped some tests based on the Hellinger distance and penalized versions of it.
\par Consider two populations $X$ and $Y$, according to classifications criteria  can be grouped into $m$ classes species    $x_1,x_2,\ldots ,x_m$  and $y_1,y_2,\ldots, y_m$ with probabilities 
 $P=(p_1,p_2,\ldots ,p_m) $ and  $ Q=(q_1,q_2,\ldots ,q_m)$ respectively. Then 
  \begin{equation}\label{eq2}  
D_\phi (P, Q)= \sum_{i=1}^{m} q_i \phi(\frac{p_i}{q_i}) 
\end{equation}

 is the $\phi-$divergence between $P$ and $Q$ (see Csisz$ \acute{a}$r, 1967) for every $\phi$ in the set $\Phi$ of real convex functions defined on $[0,\infty[$. The function $\phi(t)$ is assumed to verify the following regularity condition : \\ 
 $ \phi : [0, +\infty[\longrightarrow  \mathbb{R}\cup \{\infty \}$ is convex and continuous, where $0\phi(\frac{0}{0})=0$ and $0\phi(\frac{p}{0})=\lim_{u\longrightarrow \infty} \left( \phi(u)/u\right) $.  Its restriction on $]0,+\infty[$ is finite, twice continuously differentiable in a neighborhood of $u=1$, with  $\phi(1)=\phi'(1)=0$ and $\phi''(1)=1$ (cf. Liese and Vajda (1987)).
 \\
We shall be interested also in parametric estimators 
 \begin{equation}\label{eq3}
\widehat{Q}=\widehat{Q}_n=P_{\hat{\theta}}
\end{equation} 
of $P_{\theta_0}$  which can be obtained by means of various point estimators 
$$ \hat{\theta}=\hat{\theta}^{(n)}: \mathcal{X}^{(n)}\longrightarrow\Theta$$
of the unknown parameter $\theta_0$. \\
It is convenient to measure the difference between observed $\widehat{P}$ and expected frequencies $ P_{\theta_0}$.
 A minimum Divergence estimator of $\theta$ is a minimizer of $D_\phi (\widehat{P}, P_{\theta_0})$ where  $\widehat{P}$ is a nonparametric distribution estimate. In our case, where data come from a discrete distribution, the empirical distribution defined in \eqref{eq1} can be used.

In particular if we replace $\phi_1(x)=-4[\sqrt{x}-\frac{1}{2} (x+1)]$  in  \eqref{eq2} we get  the Hellinger distance between distribution $\widehat{P}$ and $ P_\theta$ given by 

\begin{equation}
D_{\phi_1} ( \widehat{P}, P_\theta)=HD_{\phi_1} ( \widehat{P}, P_\theta)=
\displaystyle  2\sum_{i=1}^{m} \big( \hat{p}_i^{1/2}-p_i^{1/2}(\theta)  \big)^2 \quad ; \quad \phi_1 \in \Phi. \label{eq4}
\end{equation}
Liese and Vajda (1987), Lindsay (1994) and Morales et al. (1995) introduced the so-called {\it minimum $\phi$-divergence estimate} defined by 

\begin{equation}
D_{\phi} ( \widehat{P}, P_{\widehat{\theta}})=\displaystyle  \min_{\theta \in \Theta}  D_\phi ( \widehat{P}, P_\theta) \quad ; \quad \phi \in \Phi. \label{eq5}
\end{equation}

\begin{equation}
\hat{\theta}_\phi=\displaystyle arg  \min_{\theta \in \Theta}  D_\phi ( \widehat{P}, P_\theta) \quad ; \quad \phi \in \Phi. \label{eq5}
\end{equation}
\begin{remark}
The class of estimates  \eqref{eq4} contains the maximum likelihood estimator (MLE). \\
In particular if we replace $\phi=-\log x+x-1 $ we get 
 \[  \hat{\theta}_{KL_m}=\displaystyle arg  \min_{\theta \in \Theta}  KL_m (P_\theta,\widehat{P}) =
 \displaystyle arg  \min_{\theta \in \Theta} \sum_{i=1}^m -\log p_i(\theta)\hat{p}_i = MLE \]
where $KL_m$ is  the modified Kullback-Leibler divergence.
\end{remark}
 Beran (1977) first pointed out that the minimum  Hellinger distance estimator (MHDE) of $\theta$, defined by \begin{equation}
 \hat{\theta}_H=\displaystyle arg  \min_{\theta \in \Theta}  HD_\phi ( \widehat{P}, P_\theta)  \label{eq6}
\end{equation}
 has robustness proprieties. 
 \\ Further results were given by Tamura and Boos (1986), Simpson (1987), and Donoho and Liu (1988), Simpson (1987, 1989) and Basu et al. (1997) for more details on this method of estimation. 
Simpson, however, noted that the small sample performance of the Hellinger deviance test at some discrete  models such as the Poisson is somewhat unsatisfactory, in the sense that the test requires a very large sample size for the chi-square approximation to be useful (Simpson (1989), Table 3).     
In order to avoid this problem, one possibility is to use the penalized Hellinger distance (see Harris and
Basu, (1994); Basu et al., (1996); Basu and Basu, (1998) ; Basu et al. (2010)).
The penalized Hellinger distance family between the probability vectors $\widehat{P}$ and $ P_\theta $ is defined by : 
\begin{equation}
PHD^h(\widehat{P}, P_\theta)= 2\left[ \displaystyle \sum_{i \in \varpi}^{m} \big( \hat{p}_i^{1/2}-p_i^{1/2}(\theta)  \big)^2+h\displaystyle \sum_{i \not \in \varpi^c}^{m} p_i(\theta)\right] 
\end{equation}
where  $h$ is a real positive number with 
$\varpi = \{i  : \hat{p}_i \neq 0 \}  \hbox{ and }  \varpi^c = \{i  : \hat{p}_i =  0 \}$.
Note that when $h=1$,  this generates the ordinary Hellinger distance (Simpson, 1989).\\
 Hence  \eqref{eq6} can be written as follows 
\begin{equation}
 \hat{\theta}_{PH}=\displaystyle arg  \min_{\theta \in \Theta}  PHD^h_\phi ( \widehat{P}, P_\theta)  \label{eq7}
\end{equation}
One of the suggestions to use the penalized Hellinger is motivated by the fact that this suitable choice may lead to an estimate more robust than the MLE.\\
A model selection criterion can be designed to estimate an expected overall discrepancy, a quantity  which reflects the degree of similarity between a fitted approximating model and the generating or {\it true} model. Estimation of Kullback's information (see Kullback-Leibler (1951)) is the key to deriving the Akaike Information criterion AIC (Akaike (1974)).\\
Motivated by the above developments, we propose  by analogy with
the approach introduced by Vuong (1993), a new information criterion relating to the $\phi$-divergences.  In our test,  the null hypothesis is  that the competing models are as close to the data generating process (DGP) where closeness
of a model is measured according to the discrepancy implicit in the penalized Hellinger divergence.\\
\section{Asymptotic distribution of the penalized Hellinger distance}
Hereafter, we focus on asymptotic results. We assume that the true parameter $\theta_0$ and mapping $P:\Theta\longrightarrow \Omega_m$ satisfy conditions 1-6 of Birch (1964). \\ We consider the m-vector $P_{\theta}= (p_{1 \theta},\ldots,p_{m \theta})^T$, the $m\times k$ Jacobian matrix $J_{\theta}=\left( J_{jl}(\theta)\right)_{j=1,\ldots,m;\ l=1,\ldots,k } $ with $ J_{jl}(\theta)=\displaystyle \frac{\partial}{\partial \theta_l}p_{j \theta},$ the $m \times k$ matrix $D_{\theta}=diag\left( P_{\theta}^{-1/2}\right) J_{\theta}$ and the $k\times k$ Fisher information matrix 
$$ I_{\theta}=\left( \sum_{j=1}^m \frac{1}{p_{j \theta}} \frac{\partial p_{j \theta}}{\partial \theta_r} 
 \frac{\partial p_{j \theta}}{\partial \theta_s} \right)_{r,s=1,\ldots,k}=D_{\theta}(\theta)^T D_{\theta}  $$
where $diag\left( P_{\theta}^{-1/2}\right)=diag\left( \frac{1}{\sqrt{p_1(\theta)}},\ldots,\frac{1}{\sqrt{p_m(\theta)}}\right) $\\
The above defined matrices are considered at the point $\theta \in \Theta$ where the derivatives exist and all the coordinates $p_j(\theta)$ are positive.\\

The stochastic convergences of random vectors $X_n$ to a random vector $X$ are denoted by $X_n\stackrel{P}{\longmapsto} X$ 
and  $X_n\stackrel{\mathcal{L}}{\longmapsto}  X$ (convergences in probability and in law, respectively). 
Instead $c_n X_n \stackrel{P}{\longmapsto} 0$  for a sequence of positive numbers $c_n$, we can write $\|X\|=o_p(c_n^{-1})$. \\ (This relation means $\lim_{x \rightarrow \infty} \lim \sup_{x \rightarrow \infty}
\mathbb{P}(\| c_n X_n \|> x)=0$) \\
An estimator $\widehat{P}$ of $P_{\theta_0}$ is consistent if for every $\theta_0\in \Theta$ the random vector
 $\left( \widehat{p}_1,\ldots,\widehat{p}_m\right) $ tends in probability to $\left( p_{1 \theta_0} \ldots,p_{m \theta_0}\right) $, i.e.  if
$$ \lim_{n\longrightarrow\infty}\mathbb{P}(\parallel \widehat{P} - P_{\theta_0}  \parallel > \varepsilon)=0  \hbox{  for all  }  \varepsilon >0.  $$
We need the following result to prove Theorem \eqref{th1}.
\begin{proposition} \label{prop1}(Mandal et al. 2008)\\
Let $\phi \in \Phi$, let $p: \Theta\rightarrow \Omega_m$ be twice continuously differentiable in a neighborhood of $\theta_0$ and assume that conditions 1-5 of Section 2 hold. 
Suppose that $I_{\theta_0}$ is the $k\times k$ Fisher Information matrix and $\widehat{\theta}_{PH} $ satisfying  \eqref{eq6} then the limiting distribution of $\sqrt{n}(\widehat{\theta}_{PH} - \theta_0) $ as $n \longrightarrow +\infty$ is $N[0, I^{-1}_{\theta_0}]$
\end{proposition}

\begin{lemma}
We have $$\sqrt{n}(\widehat{P}-P_{\theta_0})  \stackrel{\mathcal{L}}{\longmapsto} \mathcal{N} \left[ 0,\Sigma_{P_{\theta_0}} \right] $$
where $\widehat{P}(\theta_0)=(\widehat{p}_{1 \theta_0},\ldots,\widehat{p}_{m \theta_0}) $ an estimator of \  $ P_{\theta_0}=(p_{1 \theta_0},\ldots,p_{m \theta_0})$ defined in \eqref{eq1} \ with \\ $\Sigma_{P_{\theta_0}}=diag(P_{\theta_0})-P_{\theta_0}P_{\theta_0}^T.$
\end{lemma}
\begin{proof}
{\bf proof. }  
Denote $ \displaystyle V=\left[ \frac{N_1 -n p_{1 \theta_0}}{\sqrt{n}},\ldots, \frac{N_m -n p_{m \theta_0}}{\sqrt{n}} \right] $ \\ and 
 $N_j= \displaystyle \sum_1^n T^i_j $ where $ T^i_j(X_i) =\left\lbrace\begin{array}{lll}
1 & & \hbox{si } X_i=j\\
0& &   \hbox{otherwise}
\end{array}\right.$ 
\begin{eqnarray}
V&=&\left\lbrace \frac{1}{\sqrt{n}} \left( \sum_{i=1}^n T^i_1 - np_{1 \theta_0}\right) ; \ldots;  \frac{1}{\sqrt{n}} \left( \sum_{i=1}^n T^i_m - np_{m \theta_0}\right)  \right\rbrace  \nonumber \\
&=& \left\lbrace \sqrt{n} \left( \frac{1}{n}\sum_{i=1}^n T^i_1 - p_{1 \theta_0}\right) ; \ldots;  \frac{1}{\sqrt{n}} \left( \frac{1}{n}\sum_{i=1}^n T^i_m - p_{m \theta_0}\right)  \right\rbrace  \nonumber
\end{eqnarray}
and applying the Central Limit Theorem we have 
$$ \left( \frac{N_1 -n p_{1 \theta_0} }{\sqrt{n}},\ldots, \frac{N_m -n p_{m \theta_0}}{\sqrt{n}} \right)  \stackrel{\mathcal{L}}{\longmapsto} \mathcal{N} \left[ 0,\Sigma_{P_{\theta_0}} \right]    $$
where 
\begin{equation}\label{eq8}
\Sigma_{P_{\theta_0}}=diag(P_{\theta_0})-P_{\theta_0}P_{\theta_0}^T.
\end{equation}
\end{proof}

For simplicity, we write $D_H^h(\widehat{P},P_{\widehat{\theta}_{PH}})$ instead of $PHD^h( \widehat{P}, P_{\widehat{\theta}_{PH} })$.

\begin{theorem}
\label{th1}
Under the assumptions of Proposition \eqref{prop1}, we have 
$$ \sqrt{n}(\widehat{P}-P_{\widehat{\theta}_{PH}}) \stackrel{\mathcal{L}}{\longmapsto}  \mathcal{N} \left[ 0,\Lambda_{\theta_0}\right]  $$
where 
\begin{eqnarray}
\Lambda_{\theta_0}&=& \Sigma_{\theta_0}-\Sigma_{\theta_0} M_{\theta_0}^T - M_{\theta_0}\Sigma_{\theta_0} + M_{\theta_0}\Sigma_{\theta_0}M_{\theta_0}^T \nonumber \\ 
M_{\theta_0}&=& J_{\theta} I^{-1}_{\theta_0}(\theta_0)^T diag\big(P_{\theta_0}^{1/2}\big)\nonumber  \\
 \Sigma_{\theta_0} &=&\Sigma_{P_{\theta_0}}
\end{eqnarray}

\end{theorem}
\begin{proof}
{\bf proof. }  
A first order Taylor expansion gives 
\begin{equation}\label{eqn7} 
P_{\widehat{\theta}_{PH}}= P_{\theta_0}+ J_{\theta_0} (\widehat{\theta}_{PH}-\theta_0)^T + 
o(||\widehat{\theta}_{PH}- \theta_0 ||) 
\end{equation}
In the same way as in Morales et al. (1995), it can be established that 
:
\begin{equation}
\widehat{\theta}_{PH}= \theta_0 + I^{-1}_{\theta_0} D_{\theta_0}^T diag \left[ P_{\theta_0}^{-1/2} \right]\left( \widehat{P} - P_{\theta_0}\right) ^T + o(||\widehat{P} - P_{\theta_0} ||) \label{eqn8}
\end{equation}
From \eqref{eqn7} and \eqref{eqn8} we obtain 
$$ P_{\widehat{\theta}_{PH}}= P_{\theta_0} + J_{\theta_0}I^{-1}(\theta_0)D_{\theta_0}^T diag \left[ P_{\theta_0}^{-1/2} \right]\left( \widehat{P} - P_{\theta_0}\right) ^T + o(||\widehat{P} - P_{\theta_0} ||) $$
therefore the random vectors 
$$
\left[\begin{array}{c}
\widehat{P} - P_{\theta_0} \\
P_{\widehat{\theta}_{PH}}- P_{\theta_0} 
\end{array}
\right]_{2m\times 1} {\hbox{ and  }} 
\left[\begin{array}{c}
I \\
M_{\theta_0} 
\end{array}
\right]_{2m\times m}\times(\widehat{P} - P_{\theta_0})_{m\times 1}
$$
Where $I$ is the $m\times m$ unity matrix, have the same asymptotic distribution. \\
Furthermore it is clear (applying TCL) that 
$$
\sqrt{n}(\widehat{P}-P_{\theta_0}) \stackrel{\mathcal{L}}{\longmapsto}  \mathcal{N} \left[ 0,\Sigma_{\theta_0} \right] 
 $$
being $\Sigma_{\theta_0} $ the $m\times m$ matrix  $
diag \left[ P_{\theta_0}\right]  - P_{\theta_0}P_{\theta_0}^T $ implies 
$$ \sqrt{n}\left[\begin{array}{c}
\widehat{P} - P_{\theta_0} \\
P_{\widehat{\theta}_{PH}}- P_{\theta_0} 
\end{array}
\right]_{2m\times 1}
\stackrel{\mathcal{L}}{\longmapsto} \mathcal{N} \left[ 0,\ 
\left(\begin{array}{c}
I \\
M_{\theta_0} 
\end{array}
\right) \Sigma_{\theta_0}  (I, M_{\theta_0}^T)\right] 
   $$

therefore, we get  
\begin{equation}
\sqrt{n}(\widehat{P}-P_{\widehat{\theta}_{PH}})=\sqrt{n}(\widehat{P}-P_{{\theta_0}})+ 
\sqrt{n}(P_{\theta_0}-P_{\widehat{\theta}_{PH}}) 
 \stackrel{\mathcal{L}}{\longmapsto}  \mathcal{N} \left[ 0,\Lambda(\theta_0)\right]
\end{equation}
$$ \Lambda_{\theta_0}= \Sigma_{\theta_0}- \Sigma_{\theta_0} M_{\theta_0}^T- M_{\theta_0} \Sigma_{\theta_0}+M_{\theta_0}\Sigma_{\theta_0} M_{\theta_0}^T $$
\end{proof}

The  case which is interest to us here is to test the hypothesis $H_0 : P \in \mathcal{P}$. Our proposal is based on  the following penalized divergence test statistic $D_H^h(\widehat{P},P_{\widehat{\theta}_{PH}})$ where $\widehat{P}$ and $\widehat{\theta}_{PH}$ have been introduced in Theorem \eqref{th1} and \eqref{eq6} respectively.

Using arguments similar to those developed by Basu (1996), 
 under the assumptions of \eqref{th1} and the hypothesis $  H_0 : P=P_{\theta}$, the asymptotic distribution of  
$2nD_H^h(\widehat{P},P_{\widehat{\theta}_{PH}})$  is a chi-square when $h=1$ with $m-k-1$ degrees of freedom. Since the others members of penalized Hellinger distance tests differ from the ordinary Hellinger distance test only at the empty cells, they too have the same asymptotic distribution. (See Simpson 1989, Basu, Harris and Basu 1996 among others).

Considering now the case when the model is wrong i.e $H_1 : P\neq P_\theta$.  We introduce the following regularity assumptions

\begin{description}
\item[$(A_1)$] There exists $\theta_1=\displaystyle arg\ inf_{\theta \in \Theta} PHD^h(P,\ P_\theta)$ such that : 
$$P_{\widehat{\theta}_{PH}}\stackrel{as}{\longmapsto}P_{\theta_1} $$ 
when $ n\rightarrow +\infty $
\item[$(A_2)$] There exists $\theta_1 \in \Theta $ ; 
$ {\Lambda^\ast}=\left(  
\begin{array}{lll}
\Lambda_{11} & & \Lambda_{12}\\
\Lambda_{21} & & \Lambda_{22}
 \end{array}
  \right)$, with $\Lambda_{11}=\Sigma_p$ in \eqref{eq8}  and $\Lambda_{12} = \Lambda_{21}$ such that 
  $$ \sqrt{n} \left(  
\begin{array}{rll}
\widehat{P} & -& P \\
P_{\widehat{\theta}_{PH}} & -& P_{\theta_1}
 \end{array}
  \right) \stackrel{\mathcal{L}}{\longmapsto} \mathcal{N} \left[ 0,\Lambda^\ast \right]
  $$
\end{description}

\begin{theorem} \label{th3}
Under $H_1 : P\neq P_\theta$ and assume that conditions $(A_1)$ and $(A_2)$ hold, we have :
$$ \sqrt{n} \left(  D_H^h(\widehat{P},P_{\widehat{\theta}_{PH}}) - D_H^h(P,P_{{\theta}_{1}})\right)
 \stackrel{\mathcal{L}}{\longmapsto} \mathcal{N} \left[ 0,\Omega^2_{(\theta, P)} \right] $$
 where 
\begin{equation}
\Omega^2_{(\theta, P)}= H^T \Lambda_{11} H+H^T \Lambda_{12} J+ J^T \Lambda_{21} H+J^T \Lambda_{22} J
\end{equation} 
$H^T=(h_1,\ldots,h_m)$ with $h_i= \left(\dfrac{\partial }{\partial p_i^1} D_H^h(p^1, p^2 )\right)_{p^1=p, p^2=p(\theta_1)}  $ , 
$i=1,\ldots,m$  \\ 
and \\ $J^T=(j_1,\ldots,j_m)$ with
$j_i= \left(\dfrac{\partial }{\partial p_i^2} D_H^h(p^1, p^2 )\right)_{p^1=p, p^2=p(\theta_1)}$  , 
$i=1,\ldots,m$
\end{theorem}
\begin{proof}
{\bf proof. } 
A first order Taylor expansion gives 
  \begin{eqnarray}
  D_H^h(\widehat{P},P_{\widehat{\theta}_{PH}}) &=& D_H^h(P,P_{{\theta}_{1}}) + H^T(\widehat{P}-P)+ 
  J^T(P_{\widehat{\theta}_{PH}}-P_{\theta_1}) \nonumber \\
  & +& o(|| \widehat{P}-P||+|| P_{\widehat{\theta}_{PH}}-P_{\theta_1}||)
  \end{eqnarray}  
 From the assumed assumptions  $(A_1)$ and $(A_2)$, the result follows.
\end{proof}
\section{Applications for testing hypothesis} 
The estimate $D_H^h(\widehat{P},P_{\widehat{\theta}_{PH}})$ can be used to perform statistical tests. 
\subsection{Test of goodness-fit}
For completeness, we look at $D_H^h(\widehat{P},P_{\widehat{\theta}_{PH}})$ in the usual way, i.e., as a goodness-of-fit statistic. 
Recall that here $\theta_{PH}$ is the minimum penalized Hellinger distance estimator of $\theta$. 
Since $D_H^h(\widehat{P},P_{\widehat{\theta}_{PH}})$  is a consistent estimator of $D_H^h (P,P_{\theta })$,  the null hypothesis when using the statistic $D_H^h(\widehat{P},P_{\widehat{\theta}_{PH}})$ is 
$$ H_0 : \  D_H^h(P,P_{\theta })=0  \quad  \hbox{  or equivalently, } \quad  H_0 : \  P=P_{\theta}$$
Hence,  if $ H_0$ is rejected so that one can infer that the parametric model $P_\theta$ is misspecified. 
Since $ D_H^h(P,P_{\theta })$ is non-negative and takes value zero only when $P=P_{\theta }$, the tests are defined through the critical region 
\[ C_{\theta_{PH}}= \left\lbrace 2n D_H^h(\widehat{P},P_{\widehat{\theta}_{PH}}) > q_{\alpha,k}\right\rbrace  \]
where $q_{\alpha,k}$ is the $(1-\alpha)-$quantile of the $\chi^2-$distribution with $m-k-1$ degrees of freedom.
\begin{remark}
Theorem  \eqref{th3} can be used to give the following approximation to the power of test $ H_0 : \  D_H^h(P,P_{\theta })=0$. 
\end{remark} 
Approximated power function is 
\begin{equation}\label{eq9}
\beta_{(P)} = \mathbb{P}\left[ 2n D_H^h(\widehat{P},P_{\widehat{\theta}_{PH}}) > q_{\alpha,k} \right]  \approx 1-{\cal{F}}_{n} 
\left( \frac{q_{\alpha,k}-2n D_H^h(P,P_{\theta })}{2\sqrt{n} \Omega_{(\theta, P)} }\right) 
\end{equation}
where $q_{\alpha,k}$ is the $(1-\alpha)$-quantile of the $\chi^2$ distribution with $m-k-1$ degrees of freedom and ${\cal{F}}_{n}$ is a sequence of distribution  functions tending uniformly to the standard normal distribution ${\cal{F}}(x)$. 
Note that if $ H_0 : \  D_H^h(P,P_{\theta })\neq 0$, then for any fixed size $\alpha$ the probability of rejection $ H_0 : \  D_H^h(P,P_{\theta })=0$ with the rejection rule $2nD_H^h(\widehat{P},P_{\widehat{\theta}_{PH}}) > q_{\alpha,k} $ tends to one as $n \rightarrow\infty$. \\
Obtaining the approximate sample $n$, guaranteeing a power $\beta$ for a give alternative $P$, is an interesting application of formula \eqref{eq9}. If we wish the power to be equal to $\beta^\ast$, we must solve the equation
$$\beta^\ast =  1-{\cal{F}}\left[\frac{\sqrt{n}}{\Omega_{(\theta, P)}} \left( \frac{1}{2n}q_{\alpha,k}-  D_H^h(P,P_{\theta }) \right) \right].$$
It is not difficult to check that the sample size $n^\ast$, is the solution of the following equation
$$ n^2 D_H^h(P,P_{\theta })^2 - n D_H^h(P,P_{\theta })q_{\alpha,k} + \left( \frac{q_{\alpha,k}}{2}\right)^2 - 
n \Omega^2_{(\theta, P)} \left[  {\cal{F}}^{-1}(1-\beta^\ast)\right]^2.$$
 The solution is given by $$n^\ast=\frac{(a+b)-\sqrt{a(a+2b)}}{2D_H^h(P,P_{\theta })^2}$$
with $a=\Omega^2_{(\theta, P)} \left[{\cal{F}}^{-1}(1-\beta)\right] ^2$ and $ b=q_{\alpha,k}D_H^h(P,P_{\theta })$ and the required size is 
$ n_0= [ n^\ast ]+1$ , where $ [ \cdot] $ denotes ``integer part of".
\subsection{Test for model selection }
As we mentioned above, when one chooses a particular $\phi-$divergence type statistic $D_H^h(\widehat{P},P_{\widehat{\theta}_{PH}})=PHD_H^h(\widehat{P},P_{\widehat{\theta}_{PH}})$ with $\widehat{\theta}_{PH}$ the corresponding minimum penalized Hellinger distance estimator of $\theta$, one actually evaluates the goodness-of-fit of the parametric model $P_\theta$ according to the discrepancy $D_H^h(P,P_{\theta })$ between the true distribution $P$ and the specified model $P_\theta$.
Thus it is naturel to define the {\it best} model among a collection of competing models to be the model that is closest to the true distribution according to the discepancy $D_H^h(P,P_{\theta })$.  
\par In this paper we consider the problem of selecting between two models. Let $G_\mu=\left\lbrace G(.\mid\mu) ; \mu\in \Gamma\right\rbrace $ be another model, where $\Gamma$  is a $q-$dimensional parametric space in $\mathrm{R^q}$.
In a similar way, we can define the minimum penalized Hellinger distance estimator of $\mu$ and the corresponding discrepancy $D_H^h(P,G_{\mu })$ for the model $G_\mu$.\\
\par Our special interest is the situation in which a researcher has two competing parametric models $P_\theta$ and $G_\mu$, and he wishes to select the better of two models based on their discrimination statistic between the observations and models $P_\theta$ and $G_\mu$, defined respectively by $D_H^h(\widehat{P},P_{\widehat{\theta}_{PH}})$ and $D_H^h(\widehat{P},G_{\widehat{\mu}_{PH}})$. \\
Let the two competing parametric models $P_\theta$ and $G_\mu$ with the given discrepancy $D_H^h(P,\cdot)$. 
\begin{definition}\label{def}
\begin{eqnarray}
H_0^{eq}: & D_H^h(P,P_{\theta }) = D_H^h(P,G_{\mu })&  \hbox{ means that the two models are equivalent,}  \nonumber\\ 
H_{P_\theta} : &D_H^h(P,P_{\theta }) < D_H^h(P,G_{\mu })& \hbox{ means that  $ P_\theta$ is better than $G_\mu$,}\nonumber\\ 
H_{G_\mu} : & D_H^h(P,P_{\theta }) > D_H^h(P,G_{\mu })&  \hbox{ means that  $ P_\theta$ is worse than $G_\mu$,} \nonumber\\ 
\nonumber 
\end{eqnarray}
\end{definition}
\begin{remark} 
\par 1) It does not require that the same divergence type statistics be used in forming $D_H^h(\widehat{P},P_{\widehat{\theta}_{PH}})$ 
and $D_H^h(\widehat{P},G_{\widehat{\mu}_{PH}})$. Choosing, however, different discrepancy for evaluating competing models is hardly justified. 
\par 2) This definition does not require that either of the competing models be correctly specified. On the other hand, a correctly specified model must be at least as good as any other model. 
\end{remark}

The following expression of the indicator $D_H^h(P,P_{\theta }) - D_H^h(P,G_{\mu })$ is unknown, but from the previous section, it can be estimated by the the difference 
    $$ \sqrt{n}\left[ D_H^h(\widehat{P},P_{\widehat{\theta}_{PH}})-D_H^h(\widehat{P},G_{\widehat{\mu}_{PH}})\right] $$
This difference converges to zero under the null hypothesis $H_0^{eq}$, but converges to a strictly negative or positive constant when $H_{P_\theta}$ or $H_{G_\mu} $ holds.\\
These properties actually justify the use of $D_H^h(\widehat{P},P_{\widehat{\theta}_{PH}})-D_H^h(\widehat{P},G_{\widehat{\mu}_{PH}})$ as a model selection indicator and common procedure of selecting the model with highest goodness-of-fit. \\ 
As argued in the introduction, however, it is important to take into account the random nature of the difference $D_H^h(\widehat{P},P_{\widehat{\theta}_{PH}})-D_H^h(\widehat{P},G_{\widehat{\mu}_{PH}})$ so as to assess its significance. To do so we consider the asymptotic distribution of  $ \sqrt{n}\left[ D_H^h(\widehat{P},P_{\widehat{\theta}_{PH}})-D_H^h(\widehat{P},G_{\widehat{\mu}_{PH}})\right] $ under $H_0^{eq}$. \\
Our major task is to to propose some tests for model selection, i.e., for the null hypothesis $H_0^{eq}$ against the alternative $H_{P_\theta}$ or $H_{G_\mu} $. We use the next lemma with $\widehat{\theta}_{PH}$ and $\widehat{\mu}_{PH}$ as the corresponding  minimum penalized Hellinger distance estimator of $\theta$ and $\mu$.\\
Using $P$ and $P_\theta$ defined earlier, we consider the vector 
$$ K_\theta^T=(k_1,\ldots,k_m) \hbox{  where }  k_i= \left(\dfrac{\partial }{\partial p_i^1} D_H^h(P^1, P^2 )\right)_{P^1=P, P^2=P_{\theta}}  \hbox{  with } i=1,\dots,m $$
$$ Q_\theta^T=(q_1,\ldots,q_m) \hbox{  where }  q_i= \left(\dfrac{\partial }{\partial p_i^2} D_H^h(P^1, P^2 )\right)_{P^1=P, P^2=P_{\theta}}  \hbox{  with } i=1,\dots,m $$
\begin{lemma}\label{lem1}
Under the assumptions of the Theorem \eqref{th3}, we have
\begin{enumerate}
\item[(i)] for the model $P_\theta$,\\
 $$
  D_H^h(\widehat{P},P_{\widehat{\theta}_{PH}}) = D_H^h(P,P_{\theta}) + K_{\theta}^T(\widehat{P}-P)+ 
  Q_{\theta}^T (P_{\widehat{\theta}_{PH}}-P_{\theta}) + o_P(1) $$
\item[(ii)] for model $G_\mu$,
$$
  D_H^h(\widehat{P},G_{\widehat{\mu}_{PH}}) = D_H^h(P,G_{\mu}) + K_{\mu}^T(\widehat{P}-P)+ 
  Q_{\mu}^T (G_{\widehat{\mu}_{PH}}-G_{\mu}) + o_P(1) $$
\end{enumerate}
\end{lemma}
\begin{proof}
{\bf proof. }  \\ The results follows from a first order Taylor expansion. 
\end{proof}

We define 
$$ \Gamma^2=(K_\theta-K_\mu ;Q_{\theta}-Q_{\mu} )^T\Lambda^\ast (K_\theta-K_\mu ;Q_{\theta}-Q_{\mu} )$$ 
which is the variance of $(K_\theta-K_\mu ;Q_{\theta}-Q_{\mu} )^T  \left(  
\begin{array}{rll}
\widehat{P} & -& P \\
P_{\widehat{\theta}_{PH}} & -& P_{\theta_1}
 \end{array}
  \right) $.  
  Since $K_\theta$, $K_\mu$, $ Q_{\theta}$, $Q_{\mu}$ and $\Lambda^\ast$
are consistently estimated by their sample analogues $K_{\widehat{\theta}}$, $ K_{\widehat{\mu}}$, $Q_{\widehat{\theta}}$, $ Q_{\widehat{\mu}}$ and  
${\widehat{\Lambda}}^\ast$,  hence $\Gamma^2$  is consistently estimated by 
$$ \widehat{\Gamma}^2=(K_{\widehat{\theta}}-K_{\widehat{\mu}}; Q_{\widehat{\theta}}- Q_{\widehat{\mu}})^T 
\widehat{\Lambda}^\ast (K_{\widehat{\theta}}-K_{\widehat{\mu}}; Q_{\widehat{\theta}}- Q_{\widehat{\mu}})$$ 

Next we define the model selection statistic and its asymptotic distribution under the null and alternatives hypothesis. \\ Let 
$$
\mathcal{HI}^h=\frac{ \sqrt{n}}  {\widehat{\Gamma}} \left\lbrace 
D_H^h(\widehat{P},P_{\widehat{\theta}_{PH}})-D_H^h(\widehat{P},G_{\widehat{\mu}_{PH}})
 \right\rbrace \nonumber \\
$$
where $\mathcal{HI}^h$ stands for the penalized Hellinger Indicator.\\ 
The following theorem provides the limit distribution of $\mathcal{HI}^h$ under the null and alternatives hypothesis.
\begin{theorem}
\label{th4}
Under the assumptions of theorem \eqref{th3}, suppose that \\ $\Gamma \neq 0$, then: 
\begin{enumerate}
\item[(i)] Under the null hypothesis $H_0^{eq}$, $\mathcal{HI}^h \stackrel{\mathcal{L}}{\longmapsto}  \mathcal{N}(0,1)$
\item[(ii)] Under the null hypothesis $H_{P_\theta}$, $\mathcal{HI}^h \longrightarrow -\infty $ in probability
\item[(iii)] Under the null hypothesis $H_{G_\mu}$, $\mathcal{HI}^h \longrightarrow +\infty$ in probability
\end{enumerate}
\end{theorem}
\begin{proof}
{\bf proof. }\\
From the lemma \eqref{lem1}, it follows that

\begin{eqnarray}
D_H^h(\widehat{P},P_{\widehat{\theta}_{PH}})-  D_H^h(\widehat{P},G_{\widehat{\mu}_{PH}})&=& D_H^h(P,P_{\theta})- D_H^h(P,G_{\mu}) + K_{\theta}^T(\widehat{P}-P)- K_{\mu}^T(\widehat{P}-P) \nonumber \\ 
&+&  Q_{\theta}^T (P_{\widehat{\theta}_{PH}}-P_{\theta}) - Q_{\mu}^T (G_{\widehat{\mu}_{PH}}-G_{\mu})  + o_P(1) 
 \nonumber  
\end{eqnarray}
Under $H_0^{eq}$ : $P_{\theta}=G_{\mu}$ and $P_{\widehat{\theta}_{PH}}=G_{\widehat{\mu}_{PH}}$ we get :
\begin{eqnarray}
D_H^h(\widehat{P},P_{\widehat{\theta}_{PH}})-  D_H^h(\widehat{P},G_{\widehat{\mu}_{PH}})&=&  K_{\theta}^T(\widehat{P}-P)- K_{\mu}^T(\widehat{P}-P) \nonumber \\ 
&+&  Q_{\theta}^T (P_{\widehat{\theta}_{PH}}-P_{\theta}) - Q_{\mu}^T (P_{\widehat{\theta}_{PH}}-P_{\theta}) + o_P(1) 
 \nonumber  \\
 &=&  \left( K_{\theta}- K_{\mu}, Q_{\theta}- Q_{\mu} \right)^T  \left( \begin{array}{c}
 \widehat{P}-P \nonumber \\
 P_{\widehat{\theta}_{PH}}-P_{\theta} \nonumber \\
 \end{array}   
 \right)+ o_P(1)      
\end{eqnarray}
Finally, applying the Central Limit Theorem and assumptions (A1)-(A2), we can now immediately obtain $ \mathcal{HI}^h \stackrel{\mathcal{L}}{\longmapsto}  \mathcal{N}(0,1). $
\end{proof}
\section{ Computational results}
\subsection{Example}
To illustrate the model procedure discussed in the preceding section,we consider an example. we need to define the competing models, the estimation method used for each competing model and the Hellinger penalized type statistic to measure the departure of each proposed parametric model from the true data generating process. \\
For our competing models, we consider the problem of choosing between the family of poisson distribution and the family of geometric distribution. The poisson distribution $P(\lambda)$ is parameterized by $\lambda$ and has density
\begin{displaymath}
f(x,\lambda)=\frac{\exp({-\lambda})\times{\lambda}^x}{x!}\quad \hbox{for $x \, \in \mathbf{N}$ and  zero otherwise}.
\end{displaymath}
The geometric distribution $G(p)$ is parameterized  by  $ p $ and has density 
\begin{displaymath}
g(x,p)=(1-p)^{x-1}\times p\quad \hbox{for $x \, \in  \mathbf{N^*}$ and  zero otherwise}.
\end{displaymath} 
We use the minimum penalized Hellinger distance statistic to evaluate the discrepancy of the proposed model from the true data generating process. We partition the real line into $m$ intervals \,$\{[C_{i-1},C_i[,\,i=1,\cdots,m\}$\, where \,$C_0=0$ and $C_m=+\infty$. The choice of the cells is discussed below. \\ 
The corresponding minimum penalized Hellinger distance estimator of $\lambda$ et $p$ are :
\begin{eqnarray*}
\hat{\lambda}_{PH}=\displaystyle arg  \min_{\lambda \in \Theta}  D_H^h ( \widehat{P}, P_\lambda) 
&=& arg  \min_{\lambda \in \Theta}  \left[ \sum_{i\in \varpi}^m ({f}_i^{1/2}-p^{1/2}_{i \lambda})^2 + h\sum_{i\in \varpi^c}^m p_{i \lambda}\right] \\
\hat{p}_{PH}=\displaystyle arg  \min_{p \in \Theta}  D_H^h ( \widehat{P}, P_p) 
&=& arg  \min_{p \in \Theta}  \left[ \sum_{i\in \varpi}^m ({f}_i^{1/2}-p^{1/2}_{i p})^2 + h\sum_{i\in \varpi^c}^m p_{i p}\right] \\ 
\end{eqnarray*}
$p_{i \lambda}$ and $p_{i p}$ and are probabilities of the cells \,$[C_{i-1},C_i[$ \,under the  poisson and geometric  true distribution respectively.\\
We consider various sets of experiments in which data are generated from the mixture of a poisson and geometric distribution.\,These two distributions are calibrated so that their two means are close (4 and 5 respectively). Hence the DGP (Data Generating Process) is generated from  $M(\pi)$ with the density 
\begin{displaymath}
m(\pi)=\pi\ Pois(4)+(1-\pi)\ Geom(0.2)
\end{displaymath}    
where \,$\pi (\pi \in [0,1])$\, is specific value to  each set of experiments.
\begin{figure}
\begin{center}
 \fbox{
\includegraphics[scale=0.34]{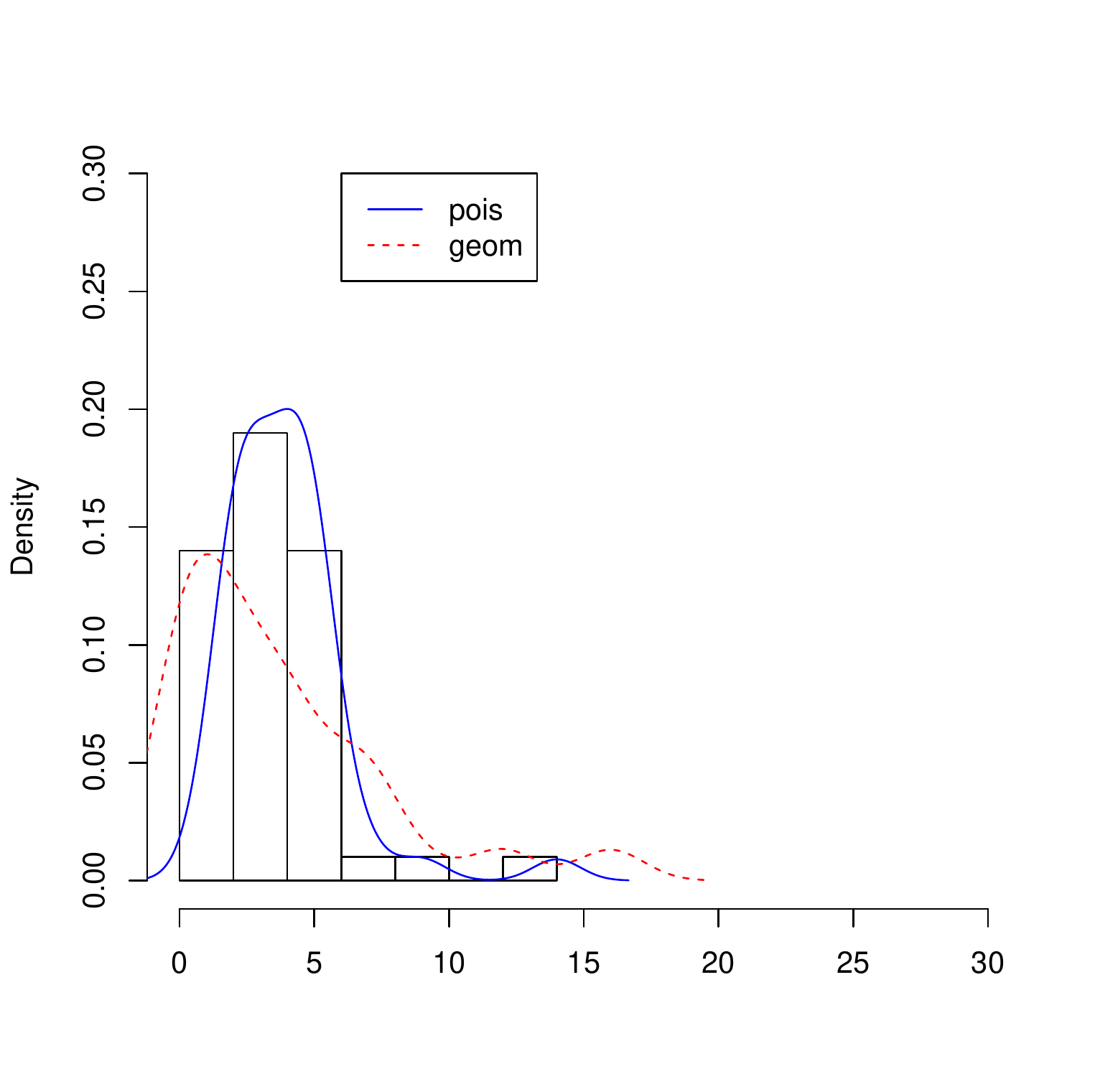}
}
\hspace{1cm}
\fbox{
\includegraphics[scale=0.34]{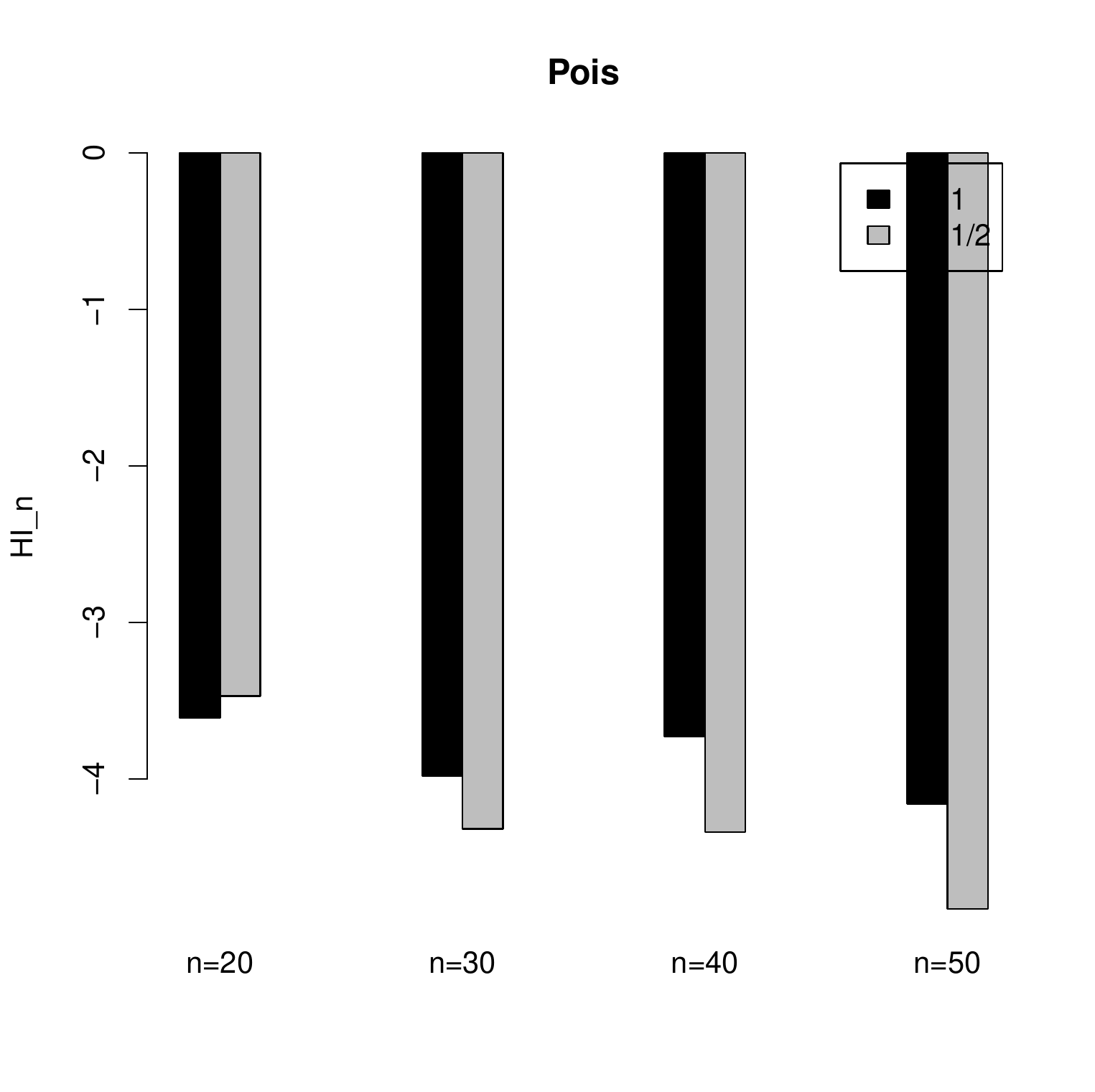}
}
\end{center}
\vspace{0.5cm}
\scriptsize {Figure 1 : Histogram of DGP=Pois(4) with n=50 } \ \ \ {\scriptsize{Figure 2 : Comparative barplot of \,$HI_n$\, depending \,$n$}}
\label{hist2}
\end{figure}
In each set of experiment several random sample are drawn from this mixture of distributions.\,The sample size varies from \,$20$\, to \,$300$,\, and for each sample size the number of replication is \,$1000$.\, In each set of experiment,\, we choose two values of the parameter \,$h=1$\, and \, $h=1/2$, where \,$h=1$\, corresponds to the classic Hellinger distance.\,The aim is to compare the accuracy of the selection model depending on the parameter setting chosen.
\begin{table}
\begin{center}
\tiny{
\begin{tabular}{|p{2cm}|p{1cm}|p{0.5cm}|p{0.5cm}|p{0.5cm}|p{0.5cm}|p{0.5cm}|p{0.5cm}|p{0.5cm}|p{0.5cm}|p{0.5cm}|p{0.5cm}|}
\hline
\multicolumn{2}{|c|}{n}&\multicolumn{2}{|c|}{20}&\multicolumn{2}{|c|}{30}&\multicolumn{2}{|c|}{40}&\multicolumn{2}{|c|}{50}&\multicolumn{2}{|c|}{300}\\
\hline
\multicolumn{2}{|c|}{$\widehat{p}$}&\multicolumn{2}{|c|}{0.210(0.03)}&\multicolumn{2}{|c|}{0.195(0.03)}&\multicolumn{2}{|c|}{0.	197(0.02)}&\multicolumn{2}{|c|}{0.205(0.02)}&\multicolumn{2}{|c|}{0.201(0.01)}\\
\hline
\multicolumn{2}{|c|}{$\widehat{\lambda}$}&\multicolumn{2}{|c|}{3.950(0.46)}&\multicolumn{2}{|c|}{4.090(0.4)}&\multicolumn{2}{|c|}{4.015(0.31)}&\multicolumn{2}{|c|}{4.015(0.28)}&\multicolumn{2}{|c|}{4.011(0.13)}\\
\hline \hline
DHP(Pois)&h=1&\multicolumn{2}{|c|}{0.133(0.07)}&\multicolumn{2}{|c|}{0.081(0.05)}&\multicolumn{2}{|c|}{0.059(0.03)}&\multicolumn{2}{|c|}{0.042(0.03)}&\multicolumn{2}{|c|}{0.037(0.01)}\\
\hline
\cline{2-10}
\multicolumn{1}{|c|}{}&h=1/2&\multicolumn{2}{|c|}{0.096(0.04)}&\multicolumn{2}{|c|}{0.064(0.03)}&\multicolumn{2}{|c|}{0.048(0.02)}&\multicolumn{2}{|c|}{0.034(0.02)}&\multicolumn{2}{|c|}{0.03(0.01)}\\
\hline
DHP(Geom)&h=1&\multicolumn{2}{|c|}{0.391(0.28)}&\multicolumn{2}{|c|}{0.348(0.12)}&\multicolumn{2}{|c|}{0.298(0.09)}&\multicolumn{2}{|c|}{0.282(0.10)}&\multicolumn{2}{|c|}{0.271(0.05)}\\
\hline
\cline{2-10}
\multicolumn{1}{|c|}{}&h=1/2&\multicolumn{2}{|c|}{0.278(0.07)}&\multicolumn{2}{|c|}{0.262(0.08)}&\multicolumn{2}{|c|}{0.242(0.06)}&\multicolumn{2}{|c|}{0.236(0.06)}&\multicolumn{2}{|c|}{0.231(0.03)}\\
\hline \hline
\multicolumn{1}{|c|}{$\mathcal{HI}^h$}&\multicolumn{1}{|c|}{$h=1/2$}&\multicolumn{2}{|c|}{-3.67(2.14)}&\multicolumn{2}{|c|}{-4.32(2.69)}&\multicolumn{2}{|c|}{-4.34(2.38)}&\multicolumn{2}{|c|}{-4.83(2.52)}&\multicolumn{2}{|c|}{-4.97(2.18)}\\
\hline
\cline{2-10}
\multicolumn{1}{|c|}{}&\multicolumn{1}{|c|}{Correct}&\multicolumn{2}{|c|}{77\%}&\multicolumn{2}{|c|}{87\%}&\multicolumn{2}{|c|}{92\%}&\multicolumn{2}{|c|}{96\%}&\multicolumn{2}{|c|}{100\%}\\
\multicolumn{1}{|c|}{}&\multicolumn{1}{|c|}{Indecisive}&\multicolumn{2}{|c|}{23\%}&\multicolumn{2}{|c|}{13\%}&\multicolumn{2}{|c|}{08\%}&\multicolumn{2}{|c|}{04\%}&\multicolumn{2}{|c|}{00\%}\\
\multicolumn{1}{|c|}{}&\multicolumn{1}{|c|}{Incorrect}&\multicolumn{2}{|c|}{00\%}&\multicolumn{2}{|c|}{00\%}&\multicolumn{2}{|c|}{00\%}&\multicolumn{2}{|c|}{00\%}&\multicolumn{2}{|c|}{00\%}\\
\hline \hline
\multicolumn{1}{|c|}{$\mathcal{HI}^h$}&\multicolumn{1}{|c|}{$h=1$}&\multicolumn{2}{|c|}{-3.61(3.03)}&\multicolumn{2}{|c|}{-3.98(2.48)}&\multicolumn{2}{|c|}{-3.73(2.29)}&\multicolumn{2}{|c|}{-4.16(2.35)}&\multicolumn{2}{|c|}{-4.25(1.87)}\\
\hline
\cline{2-10}
\multicolumn{1}{|c|}{}&\multicolumn{1}{|c|}{Correct}&\multicolumn{2}{|c|}{70\%}&\multicolumn{2}{|c|}{79\%}&\multicolumn{2}{|c|}{83\%}&\multicolumn{2}{|c|}{86\%}&\multicolumn{2}{|c|}{93\%}\\
\multicolumn{1}{|c|}{}&\multicolumn{1}{|c|}{Indecisive}&\multicolumn{2}{|c|}{30\%}&\multicolumn{2}{|c|}{21\%}&\multicolumn{2}{|c|}{17\%}&\multicolumn{2}{|c|}{14\%}&\multicolumn{2}{|c|}{07\%}\\
\multicolumn{1}{|c|}{}&\multicolumn{1}{|c|}{Incorrect}&\multicolumn{2}{|c|}{00\%}&\multicolumn{2}{|c|}{00\%}&\multicolumn{2}{|c|}{00\%}&\multicolumn{2}{|c|}{00\%}&\multicolumn{2}{|c|}{00\%}\\
\hline
\multicolumn{12}{c}{\small{\textit{\scriptsize {\ Table 1 : DGP=Pois(4)}}}}
\end{tabular}}
\end{center}
\end{table}
In order a perfect fit by the proposed method, for the chosen parameters of these two distributions, we note that most of the mass is concentrated between 0 and 10. Therefore, the chosen partition has eight cells defined by \,$\{[C_{i-1}, C_i[=[i-1, i[,\,i=1,\cdots,7\}$\ and  $[C_7, C_8[=[7, +\infty[$ represents the last cell. 
We choose different values of \,$\pi$\, which are \,$0.00,\,0.25, 0.535,\,0.75,\,1.00$.\,Although our proposed model selection procedure does not require that the data generating process belong to either of the competing models,\, we consider the two limiting cases \,$\pi=1.00$\, and \,$\pi=0.00$\, for they correspond to the correctly specified cases.\,To investigate the case where both competing models are misspecified but not at equal distance from the DGP,\, we consider the case  $\pi=0.25$, $\pi=0.75$ and $\pi=0.535$. The former case correspond to a DGP  which is poisson but slightly contaminated by a geometric distribution.\,The second case is interpreted similarly as a geometric  slightly contaminated by a poisson distribution. 
In the last case,  $\pi=0.535$ is the value for which the poisson $D_H^h(\widehat{P},P_{\widehat{\lambda}_{PH}})$ and the geometric $D_H^h(\widehat{P},G_{\widehat{p}_{PH}})$ families are approximatively at equal distance to the mixture $m(\pi)$ according to the penalized Hellinger distance with the above cells. Thus this set of experiments corresponds approximatively to the null hypothesis of our proposed model selection test  $\mathcal{HI}^h$. 
\begin{figure}
\begin{center}
\fbox{

\includegraphics[scale=0.34]{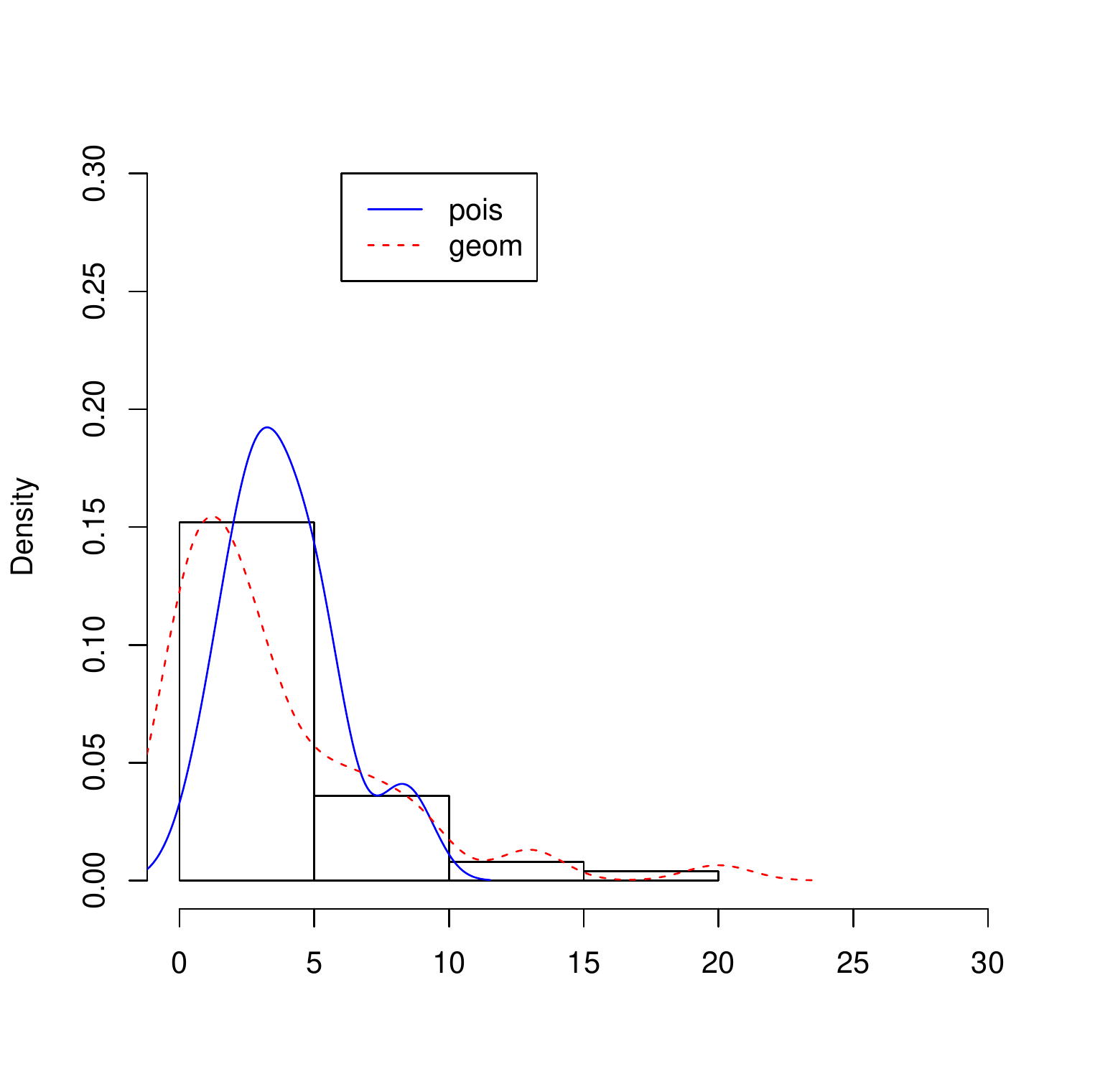}
}
\hspace{1cm}
\fbox{
\includegraphics[scale=0.34]{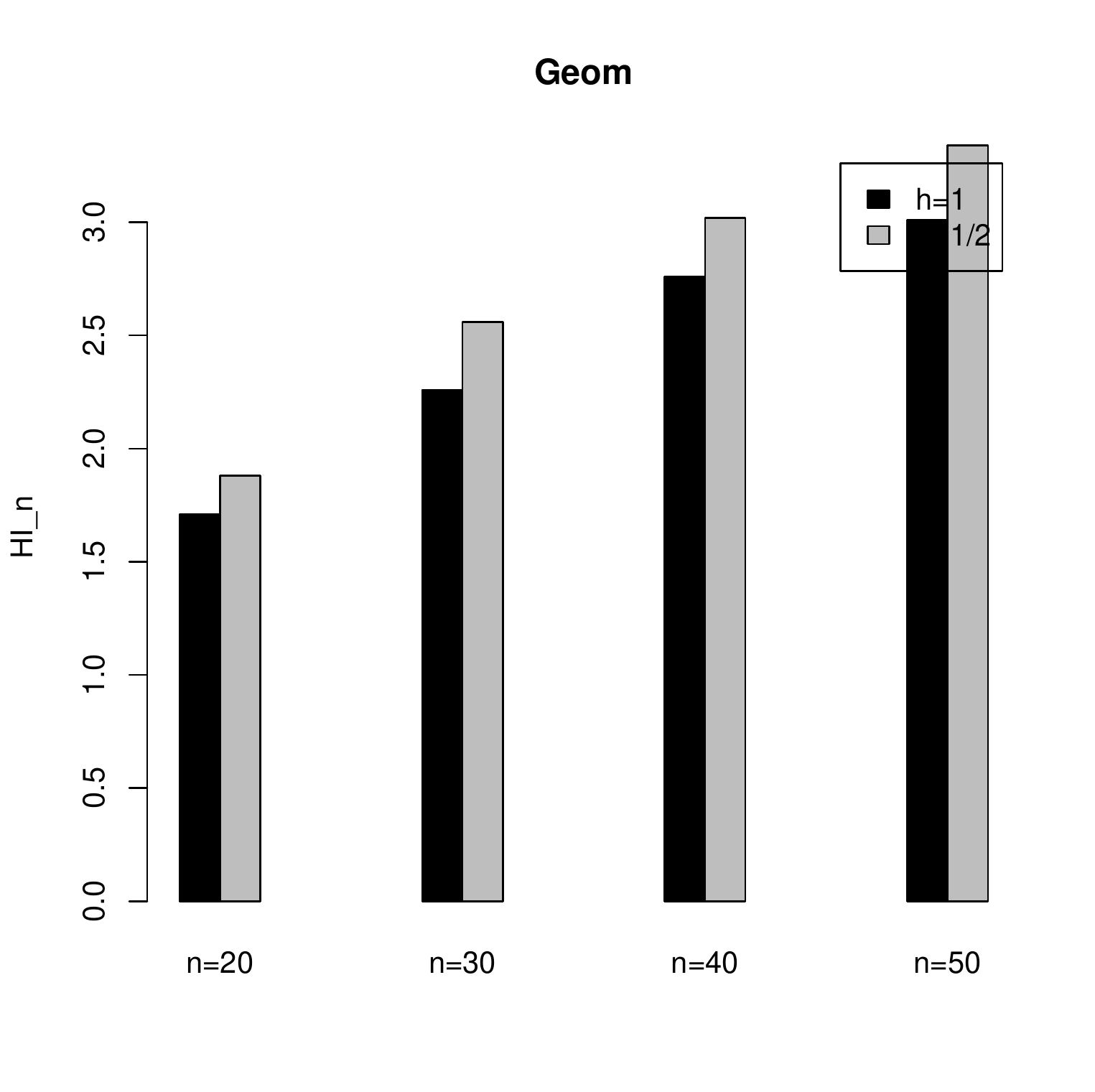}
}
\end{center}
\vspace{0.5cm}
\scriptsize {Figure 3 : Histogram of DGP=Geom(0.2) with n=50 } \ {\scriptsize{Figure 4 : Comparative barplot of \,$HI_n$\, depending,$n$}}
\label{stat3}
\end{figure}
The results of our different sets of experiments are presented in table 1-5.\, The first half of each table gives the average values of the the minimum penalized Hellinger distance estimator\,$\widehat{\lambda}_{PH}$\, and \,$\widehat{p}_{PH}$,\,the penalized Hellinger goodness-of-fit statistics \,$D_H^h(\widehat{P},P_{\widehat{\lambda}_{PH}})$\, and \,$D_H^h(\widehat{P},G_{\widehat{p}_{PH}})$,\, and the Hellinger indicator statistic \,$\mathcal{HI}^h$. The values in parentheses are standard errors. The second half of each table gives in percentage the number of times our proposed model selection procedure based on $\mathcal{HI}^h$  favors the poisson model,\,the geometric model,\,and indecisive.\,The tests are conducted at \,$5\%$\, nominal significance level.

\begin{table}
\begin{center}
\tiny{
\begin{tabular}{|p{2cm}|p{1cm}|p{0.5cm}|p{0.5cm}|p{0.5cm}|p{0.5cm}|p{0.5cm}|p{0.5cm}|p{0.5cm}|p{0.5cm}|p{0.5cm}|p{0.5cm}|}
\hline
\multicolumn{2}{|c|}{n}&\multicolumn{2}{|c|}{20}&\multicolumn{2}{|c|}{30}&\multicolumn{2}{|c|}{40}&\multicolumn{2}{|c|}{50}&\multicolumn{2}{|c|}{300}\\
\hline
\multicolumn{2}{|c|}{$\widehat{p}$}&\multicolumn{2}{|c|}{0.196(0.04)}&\multicolumn{2}{|c|}{0.213(0.03)}&\multicolumn{2}{|c|}{0.203(0.02)}&\multicolumn{2}{|c|}{0.203(0.02)}&\multicolumn{2}{|c|}{0.201(0,01)}\\
\hline
\multicolumn{2}{|c|}{$\widehat{\lambda}$}&\multicolumn{2}{|c|}{3.920(1.0)}&\multicolumn{2}{|c|}{4.206(0.89)}&\multicolumn{2}{|c|}{4.021(0.67)}&\multicolumn{2}{|c|}{4.109(0.58)}&\multicolumn{2}{|c|}{4.03(0.34)}\\
\hline \hline
DHP(Pois)&h=1.0&\multicolumn{2}{|c|}{0.356(0.14)}&\multicolumn{2}{|c|}{0.309(0.10)}&\multicolumn{2}{|c|}{0.271(0.09)}&\multicolumn{2}{|c|}{0.253(0.08)}&\multicolumn{2}{|c|}{0.244(0.07)}\\
\hline
\cline{2-10}
\multicolumn{1}{|c|}{}&h=0.5&\multicolumn{2}{|c|}{0.281(0.1)}&\multicolumn{2}{|c|}{0.273(0.07)}&\multicolumn{2}{|c|}{0.254(0.07)}&\multicolumn{2}{|c|}{0.246(0.07)}&\multicolumn{2}{|c|}{0.237(0.02)}\\
\hline
DHP(Geom)&h=1&\multicolumn{2}{|c|}{0.150(0.06)}&\multicolumn{2}{|c|}{0.089(0.05)}&\multicolumn{2}{|c|}{0.053(0.03)}&\multicolumn{2}{|c|}{0.039(0.02)}&\multicolumn{2}{|c|}{0.033(0.01)}\\
\cline{2-10}
\hline
\multicolumn{1}{|c|}{}&h=1/2&\multicolumn{2}{|c|}{0.103(0.04)}&\multicolumn{2}{|c|}{0.067(0.03)}&\multicolumn{2}{|c|}{0.044(0.02)}&\multicolumn{2}{|c|}{0.035(0.02)}&\multicolumn{2}{|c|}{0.027(0.98)}\\
\hline \hline
\multicolumn{1}{|c|}{$\mathcal{HI}^h$}&\multicolumn{1}{|c|}{$h=1/2$}&\multicolumn{2}{|c|}{1.880(1.43)}&\multicolumn{2}{|c|}{2.560(1.37)}&\multicolumn{2}{|c|}{3.020(1.25)}&\multicolumn{2}{|c|}{3.340(1.14)}&\multicolumn{2}{|c|}{3.40(1.03)}\\ \hline
\cline{2-10}
\multicolumn{1}{|c|}{}&\multicolumn{1}{|c|}{Correct}&\multicolumn{2}{|c|}{42\%}&\multicolumn{2}{|c|}{72\%}&\multicolumn{2}{|c|}{81\%}&\multicolumn{2}{|c|}{90\%}&\multicolumn{2}{|c|}{97\%}\\
\multicolumn{1}{|c|}{}&\multicolumn{1}{|c|}{Indecisive}&\multicolumn{2}{|c|}{58\%}&\multicolumn{2}{|c|}{28\%}&\multicolumn{2}{|c|}{19\%}&\multicolumn{2}{|c|}{10\%}&\multicolumn{2}{|c|}{03\%}\\
\multicolumn{1}{|c|}{}&\multicolumn{1}{|c|}{Incorrect}&\multicolumn{2}{|c|}{00\%}&\multicolumn{2}{|c|}{00\%}&\multicolumn{2}{|c|}{00\%}&\multicolumn{2}{|c|}{00\%}&\multicolumn{2}{|c|}{00\%}\\
\hline \hline
\multicolumn{1}{|c|}{$\mathcal{HI}^h$}&\multicolumn{1}{|c|}{$h=1$}&\multicolumn{2}{|c|}{1.710(1.07)}&\multicolumn{2}{|c|}{2.260(1.05)}&\multicolumn{2}{|c|}{2.760(0.96)}&\multicolumn{2}{|c|}{3.01(0.65)}&\multicolumn{2}{|c|}{4.19(0.32)}\\ \hline
\cline{2-10}
\multicolumn{1}{|c|}{}&\multicolumn{1}{|c|}{Correct}&\multicolumn{2}{|c|}{36\%}&\multicolumn{2}{|c|}{62\%}&\multicolumn{2}{|c|}{77\%}&\multicolumn{2}{|c|}{84\%}&\multicolumn{2}{|c|}{92\%}\\
\multicolumn{1}{|c|}{}&\multicolumn{1}{|c|}{Indecisive}&\multicolumn{2}{|c|}{64\%}&\multicolumn{2}{|c|}{38\%}&\multicolumn{2}{|c|}{23\%}&\multicolumn{2}{|c|}{16\%}&\multicolumn{2}{|c|}{08\%}\\
\multicolumn{1}{|c|}{}&\multicolumn{1}{|c|}{Incorrect}&\multicolumn{2}{|c|}{00\%}&\multicolumn{2}{|c|}{00\%}&\multicolumn{2}{|c|}{00\%}&\multicolumn{2}{|c|}{00\%}&\multicolumn{2}{|c|}{00\%}\\
\hline
\multicolumn{10}{c}{\small{\textit{\scriptsize {\ Table 2 : DGP=Geom(0.2)}}}}
\end{tabular}}
\label{tab1}
\end{center}
\end{table}

 In the first two sets of experiments ($\pi=0.00 \hbox{ and } \pi=1.00$) where one model is correctly specified, we use the labels {\it `correct'}, {\it `incorrect'} and {\it `indecisive'} when a choice is made. 
\begin{figure}
\begin{center}
 \fbox{
\includegraphics[scale=0.34]{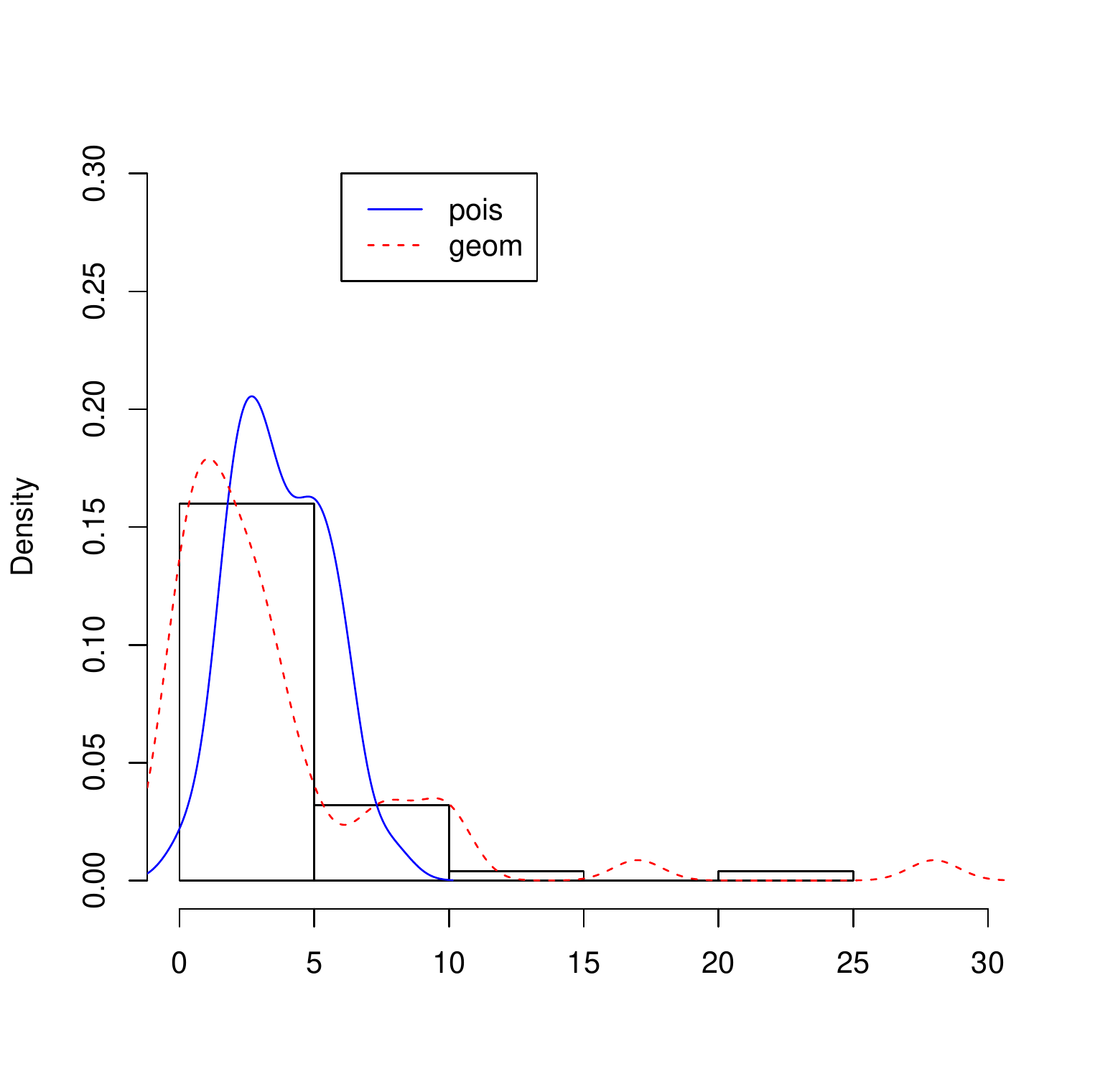}
}
\hspace{1cm}
\fbox{
\includegraphics[scale=0.34]{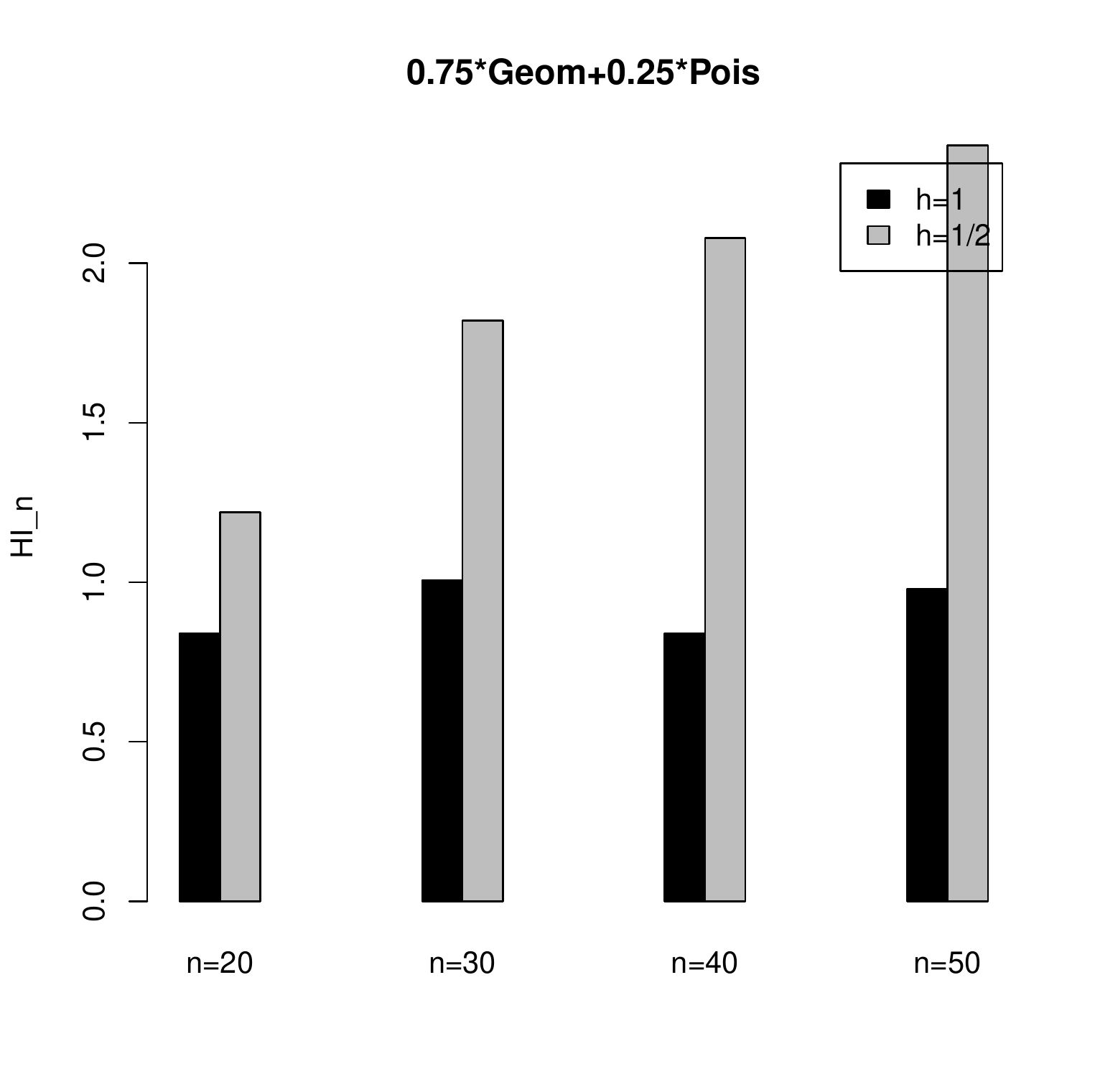}
}
\end{center}
\vspace{0.5cm}
\scriptsize {Figure 5 : Histogram of \\ DGP=0.75$\times$Geom+0.25$\times$Pois with n=50 } \hspace{1.3cm} {\scriptsize{Figure 6 : Comparative barplot of \,$HI_n$\, depending \,$n$ \\ }}
\label{hist3}
\end{figure}
The first halves of tables 1-5 confirm our asymptotic results. They all show that the minimum penalized Hellinger estimators $\widehat{\lambda}_{PH}$ and $\widehat{p}_{_{PH}}$ converge to their pseudo-true values in the misspecified cases and to their true values in the correctly specified cases as the sample size increases . With respect to our $\mathcal{HI}^h$, it diverges to $-\infty$ or $+\infty$ at the approximate rate of $\sqrt{n}$  except in the table 5. In the latter case the $\mathcal{HI}^h$ statistic converges, as expected, to zero which is the mean of the asymptotic $\mathcal{N}(0,1)$ distribution  under our null hypothesis of equivalence.   
\begin{table}
\begin{center}
\tiny{
\begin{tabular}{|p{2cm}|p{1cm}|p{0.5cm}|p{0.5cm}|p{0.5cm}|p{0.5cm}|p{0.5cm}|p{0.5cm}|p{0.5cm}|p{0.5cm}|p{0.5cm}|p{0.5cm}|}
\hline
\multicolumn{2}{|c|}{n}&\multicolumn{2}{|c|}{20}&\multicolumn{2}{|c|}{30}&\multicolumn{2}{|c|}{40}&\multicolumn{2}{|c|}{50}&\multicolumn{2}{|c|}{300}\\
\hline
\multicolumn{2}{|c|}{$\widehat{p}$}&\multicolumn{2}{|c|}{0.213(0.13)}&\multicolumn{2}{|c|}{0.197(0.12)}&\multicolumn{2}{|c|}{0.208(0.08)}&\multicolumn{2}{|c|}{0.202(0.05)}&\multicolumn{2}{|c|}{0.202(0.01)}\\
\hline
\multicolumn{2}{|c|}{$\widehat{\lambda}$}&\multicolumn{2}{|c|}{4.160(0.72)}&\multicolumn{2}{|c|}{3.910(0.55)}&\multicolumn{2}{|c|}{4.180(0.55)}&\multicolumn{2}{|c|}{3.970(0.43)}&\multicolumn{2}{|c|}{4.022(0.21)}\\
\hline \hline
DHP(Pois)&h=1&\multicolumn{2}{|c|}{0.546(0.13)}&\multicolumn{2}{|c|}{0.472(0.1)}&\multicolumn{2}{|c|}{0.412(0.09)}&\multicolumn{2}{|c|}{0.402(0.08)}&\multicolumn{2}{|c|}{0.367(0.06)} \\ \hline
\cline{2-10}
\multicolumn{1}{|c|}{}&h=1/2&\multicolumn{2}{|c|}{0.344(0.07)}&\multicolumn{2}{|c|}{0.340(0.05)}&\multicolumn{2}{|c|}{0.320(0.05)}&\multicolumn{2}{|c|}{0.311(0.05)}&\multicolumn{2}{|c|}{0.304(0.03)}\\
DHP(Geom)&h=1&\multicolumn{2}{|c|}{0.150(0.06)}&\multicolumn{2}{|c|}{0.089(0.05)}&\multicolumn{2}{|c|}{0.053(0.03)}&\multicolumn{2}{|c|}{0.039(0.02)}&\multicolumn{2}{|c|}{0.021(0.01)}\\ \hline
\cline{2-10}
\multicolumn{1}{|c|}{}&h=1/2&\multicolumn{2}{|c|}{-3.67(2.62)}&\multicolumn{2}{|c|}{-4.32(2.53)}&\multicolumn{2}{|c|}{-4.34(2.47)}&\multicolumn{2}{|c|}{-4.83(2.27)}&\multicolumn{2}{|c|}{-5.37(2.01)}\\
\hline \hline
\multicolumn{1}{|c|}{$\mathcal{HI}^h$}&\multicolumn{1}{|c|}{$h=1/2$}&\multicolumn{2}{|c|}{1.220(1.02)}&\multicolumn{2}{|c|}{1.820(0.89)}&\multicolumn{2}{|c|}{2.080(1.12)}&\multicolumn{2}{|c|}{2.370(0.99)}&\multicolumn{2}{|c|}{3.102(0.84)}\\ \hline
\cline{2-10}
\multicolumn{1}{|c|}{}&\multicolumn{1}{|c|}{Geom}&\multicolumn{2}{|c|}{23\%}&\multicolumn{2}{|c|}{40\%}&\multicolumn{2}{|c|}{50\%}&\multicolumn{2}{|c|}{64\%}&\multicolumn{2}{|c|}{81\%}\\
\multicolumn{1}{|c|}{}&\multicolumn{1}{|c|}{Indecisive}&\multicolumn{2}{|c|}{77\%}&\multicolumn{2}{|c|}{60\%}&\multicolumn{2}{|c|}{50\%}&\multicolumn{2}{|c|}{36\%}&\multicolumn{2}{|c|}{19\%}\\
\multicolumn{1}{|c|}{}&\multicolumn{1}{|c|}{Pois}&\multicolumn{2}{|c|}{00\%}&\multicolumn{2}{|c|}{00\%}&\multicolumn{2}{|c|}{00\%}&\multicolumn{2}{|c|}{00\%}&\multicolumn{2}{|c|}{00\%}\\
\hline \hline
\multicolumn{1}{|c|}{$\mathcal{HI}^h$}&\multicolumn{1}{|c|}{$h=1$}&\multicolumn{2}{|c|}{0.840(1.29)}&\multicolumn{2}{|c|}{0.831(1.27)}&\multicolumn{2}{|c|}{0.845(1.16)}&\multicolumn{2}{|c|}{0.967(1.05)}&\multicolumn{2}{|c|}{1.131(0.78)}\\ \hline
\cline{2-10}
\multicolumn{1}{|c|}{}&\multicolumn{1}{|c|}{Geom}&\multicolumn{2}{|c|}{17\%}&\multicolumn{2}{|c|}{15\%}&\multicolumn{2}{|c|}{19\%}&\multicolumn{2}{|c|}{22\%}&\multicolumn{2}{|c|}{33\%}\\
\multicolumn{1}{|c|}{}&\multicolumn{1}{|c|}{Indecisive}&\multicolumn{2}{|c|}{80\%}&\multicolumn{2}{|c|}{83\%}&\multicolumn{2}{|c|}{89\%}&\multicolumn{2}{|c|}{77\%}&\multicolumn{2}{|c|}{66\%}\\
\multicolumn{1}{|c|}{}&\multicolumn{1}{|c|}{Pois}&\multicolumn{2}{|c|}{03\%}&\multicolumn{2}{|c|}{02\%}&\multicolumn{2}{|c|}{02\%}&\multicolumn{2}{|c|}{01\%}&\multicolumn{2}{|c|}{01\%}\\
\hline
\multicolumn{10}{c}{\small{\textit{\scriptsize {\ Table 3 : DGP=0.75$\times$Geom(0.2)+0.25$\times$Pois(4)}}}}
\end{tabular}}
\label{tab3}
\end{center}
\end{table}

With the exception of table 1 and 2, we observed a large percentage of incorrect decisions. This is because both models are now incorrectly specified. In contrast, turning to the second halves of the tables 1-2, we first note  that the  percentage of correct choices using  $\mathcal{HI}^h$ statistic steadily increases and ultimately converges to $100\%$.

\begin{figure}
\begin{center}
 \fbox{
\includegraphics[scale=0.34]{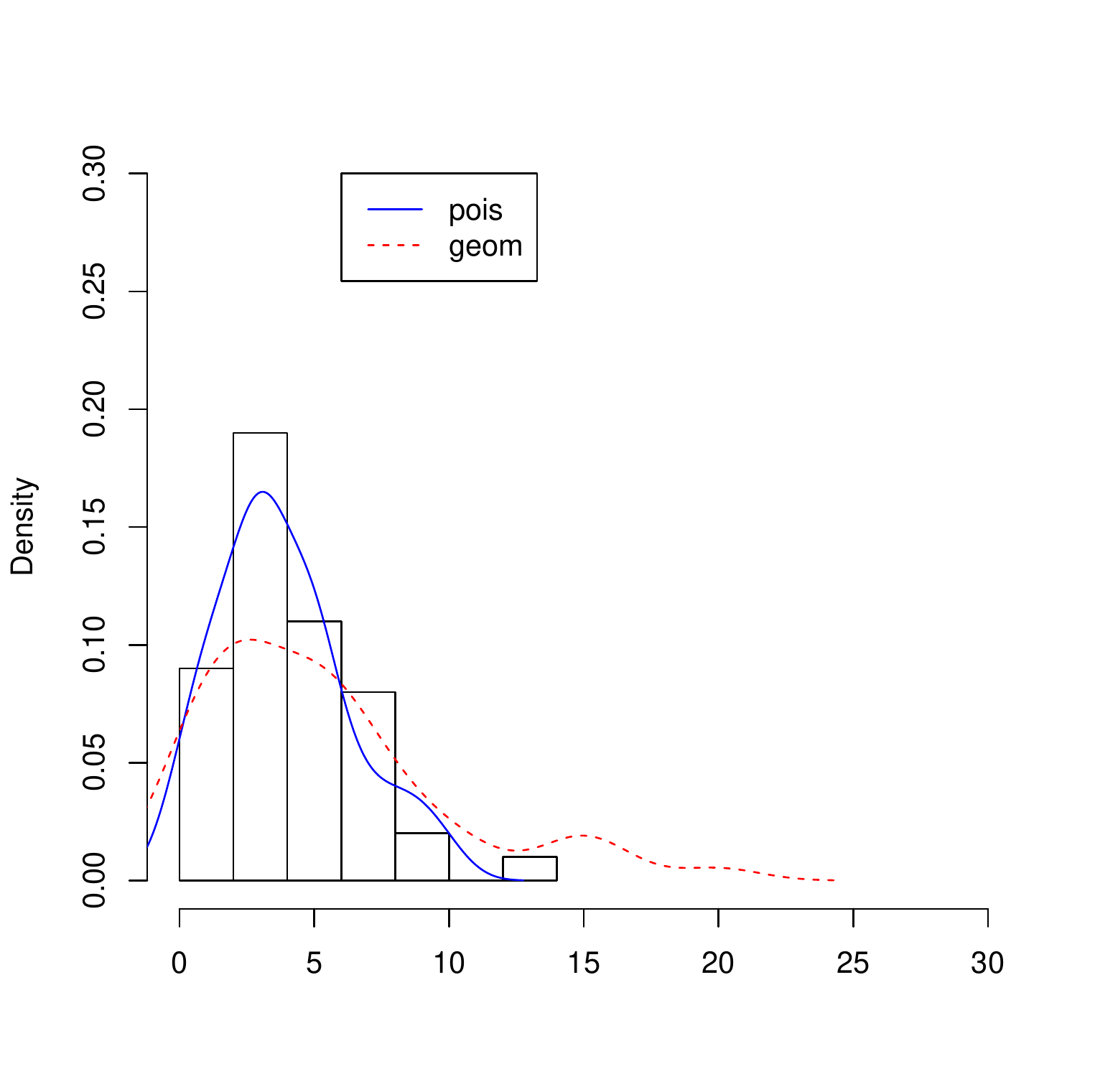}
}
\hspace{1cm}
\fbox{
\includegraphics[scale=0.34]{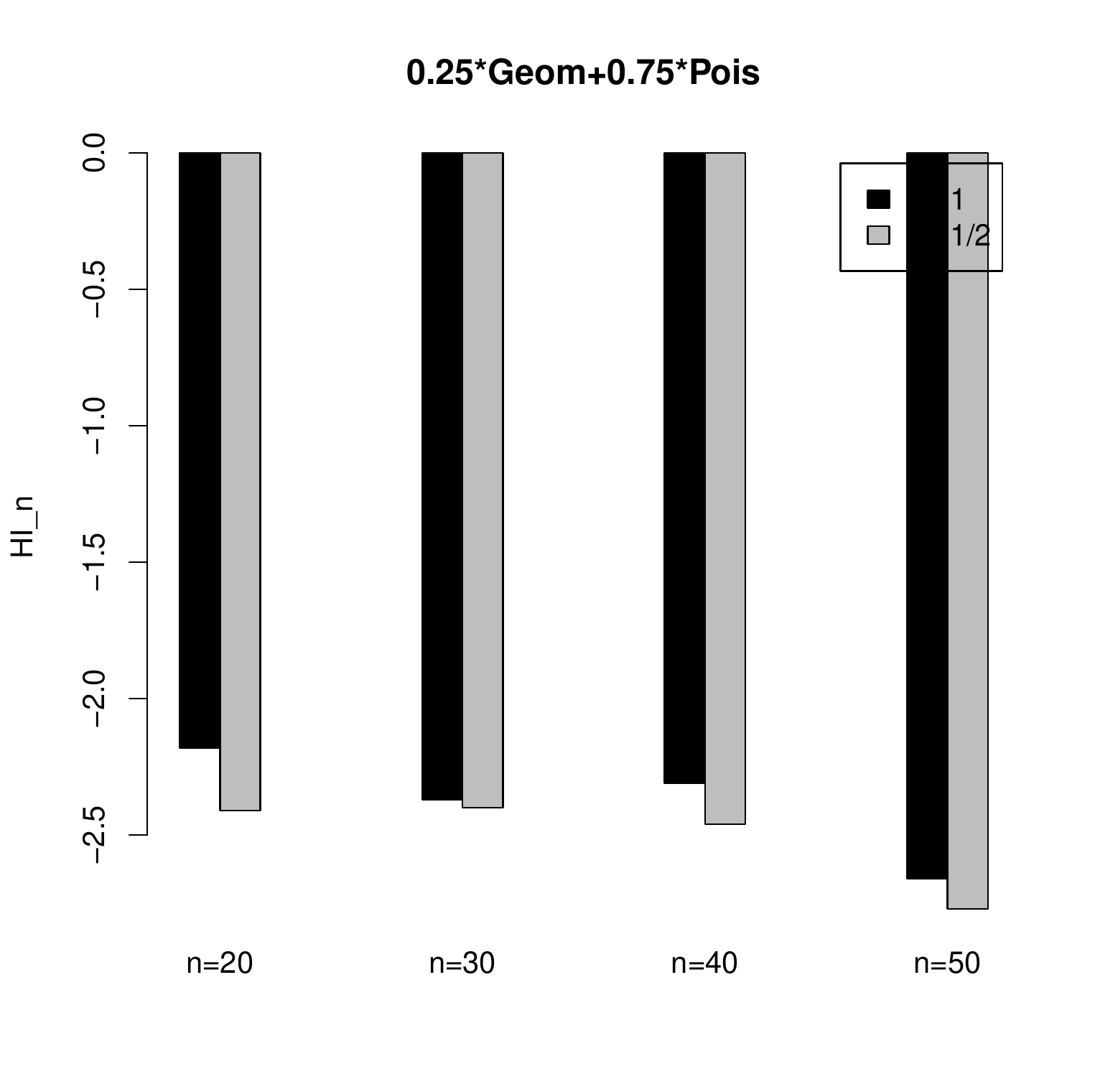}
}
\end{center}
\vspace{0.5cm}
\scriptsize {Figure 7 : Histogram of \\ DGP=0.25$\times$Geom+0.75$\times$Pois with n=50 } \hspace{0.5cm} {\scriptsize{Figure 8 : Comparative barplot of \,$HI_n$\, depending \,$n$ \\ }}

\label{hist4}
\end{figure}

The preceding comments for the second halves of tables 1 and 2 also apply to the second halves of 
tables 3 and 4. In all  tables (1,2,3 and 4), the results confirm, in small samples, the relative domination of the model selection procedure based on the penalized Hellinger statistic test ($h=1/2$) than the other corresponding to the choice  of classical Hellinger statistic test ($h=1$), in percentages of correct decisions. Table 5 also confirms our asymptotics results : as sample size incerases, the percentage of rejection of both models converges, as it should, to 100\%. 

\begin{table}
\begin{center}
\tiny{
\begin{tabular}{|p{2cm}|p{1cm}|p{0.5cm}|p{0.5cm}|p{0.5cm}|p{0.5cm}|p{0.5cm}|p{0.5cm}|p{0.5cm}|p{0.5cm}|p{0.5cm}|p{0.5cm}|}
\hline
\multicolumn{2}{|c|}{n}&\multicolumn{2}{|c|}{20}&\multicolumn{2}{|c|}{30}&\multicolumn{2}{|c|}{40}&\multicolumn{2}{|c|}{50}&\multicolumn{2}{|c|}{300}\\
\hline
\multicolumn{2}{|c|}{$\widehat{p}$}&\multicolumn{2}{|c|}{0.213(0.03)}&\multicolumn{2}{|c|}{0.212(0.03)}&\multicolumn{2}{|c|}{0.210(0.02)}&\multicolumn{2}{|c|}{0.206(0.02)}&\multicolumn{2}{|c|}{0.203(0.01)}\\
\hline
\multicolumn{2}{|c|}{$\widehat{\lambda}$}&\multicolumn{2}{|c|}{4.110(0.43)}&\multicolumn{2}{|c|}{4.090(0.31)}&\multicolumn{2}{|c|}{3.970(0.28)}&\multicolumn{2}{|c|}{4.020(0.26)}&\multicolumn{2}{|c|}{4.019(0.17)}\\ \hline
\hline 
DHP(Pois)&h=1&\multicolumn{2}{|c|}{1.779(0.45)}&\multicolumn{2}{|c|}{1.634(0.30)}&\multicolumn{2}{|c|}{1.650(0.28)}&\multicolumn{2}{|c|}{1.570(0.24)}&\multicolumn{2}{|c|}{1.520(0.21)}\\ \hline
\cline{2-10}
\multicolumn{1}{|c|}{}&h=1/2&\multicolumn{2}{|c|}{1.443(0.24)}&\multicolumn{2}{|c|}{1.473(0.21)}&\multicolumn{2}{|c|}{1.520(0.20)}&\multicolumn{2}{|c|}{1.500(0.18)}&\multicolumn{2}{|c|}{1.483(0.14)}\\ \hline
DHP(Geom)&h=1&\multicolumn{2}{|c|}{2.055(0.35)}&\multicolumn{2}{|c|}{1.870(0.25)}&\multicolumn{2}{|c|}{1.860(0.21)}&\multicolumn{2}{|c|}{1.790(0.19)}&\multicolumn{2}{|c|}{1.704(0.11)}\\ \hline
\cline{2-10}
\multicolumn{1}{|c|}{}&h=1/2&\multicolumn{2}{|c|}{1.640(0.15)}&\multicolumn{2}{|c|}{1.660(0.15)}&\multicolumn{2}{|c|}{1.700(0.14)}&\multicolumn{2}{|c|}{1.690(0.13)}&\multicolumn{2}{|c|}{1.632(0.10)}\\ \hline
\hline 
\multicolumn{1}{|c|}{$\mathcal{HI}^h$}&\multicolumn{1}{|c|}{$h=1/2$}&\multicolumn{2}{|c|}{-2.40(1.27)}&\multicolumn{2}{|c|}{-2.44(1.1)}&\multicolumn{2}{|c|}{-2.49(1.08)}&\multicolumn{2}{|c|}{-2.77(1.01)}&\multicolumn{2}{|c|}{-2.89(0.92)}\\ \hline
\cline{2-10}
\multicolumn{1}{|c|}{}&\multicolumn{1}{|c|}{Geom}&\multicolumn{2}{|c|}{00\%}&\multicolumn{2}{|c|}{00\%}&\multicolumn{2}{|c|}{00\%}&\multicolumn{2}{|c|}{00\%}&\multicolumn{2}{|c|}{00\%}\\
\multicolumn{1}{|c|}{}&\multicolumn{1}{|c|}{Indecisive}&\multicolumn{2}{|c|}{38\%}&\multicolumn{2}{|c|}{37\%}&\multicolumn{2}{|c|}{32\%}&\multicolumn{2}{|c|}{27\%}&\multicolumn{2}{|c|}{21\%}\\
\multicolumn{1}{|c|}{}&\multicolumn{1}{|c|}{Pois}&\multicolumn{2}{|c|}{62\%}&\multicolumn{2}{|c|}{63\%}&\multicolumn{2}{|c|}{68\%}&\multicolumn{2}{|c|}{83\%}&\multicolumn{2}{|c|}{79\%}\\
\hline \hline
\multicolumn{1}{|c|}{$\mathcal{HI}^h$}&\multicolumn{1}{|c|}{$h=1$}&\multicolumn{2}{|c|}{-2.18(1.37)}&\multicolumn{2}{|c|}{-2.37(1.33)}&\multicolumn{2}{|c|}{-2.31(1.36)}&\multicolumn{2}{|c|}{-2.66(1.18)}&\multicolumn{2}{|c|}{-2.83(1.06)}\\ \hline
\cline{2-10}
\multicolumn{1}{|c|}{}&\multicolumn{1}{|c|}{Geom}&\multicolumn{2}{|c|}{00\%}&\multicolumn{2}{|c|}{00\%}&\multicolumn{2}{|c|}{00\%}&\multicolumn{2}{|c|}{00\%}&\multicolumn{2}{|c|}{00\%}\\
\multicolumn{1}{|c|}{}&\multicolumn{1}{|c|}{Indecisive}&\multicolumn{2}{|c|}{48\%}&\multicolumn{2}{|c|}{45\%}&\multicolumn{2}{|c|}{46\%}&\multicolumn{2}{|c|}{30\%}&\multicolumn{2}{|c|}{24\%}\\
\multicolumn{1}{|c|}{}&\multicolumn{1}{|c|}{Pois}&\multicolumn{2}{|c|}{52\%}&\multicolumn{2}{|c|}{55\%}&\multicolumn{2}{|c|}{54\%}&\multicolumn{2}{|c|}{70\%}&\multicolumn{2}{|c|}{76\%}\\
\hline
\multicolumn{10}{c}{\small{\textit{\scriptsize {\ Table 4 : DGP=0.75$\times$Pois(4)+0.25$\times$Geom(0.2)}}}}
\end{tabular}}
\label{tab4}
\end{center}
\end{table}

In figures 1, 3, 5, 7 and 9 we plot the histogramm of datasets and overlay the curves for 
Geometric and poisson distribution. When the DGP is correctly specified figure 1, the poisson 
distribution  has a reasonable chance of being distinguished from geometric distribution. 

\begin{figure}
\begin{center}
 \fbox{
\includegraphics[scale=0.34]{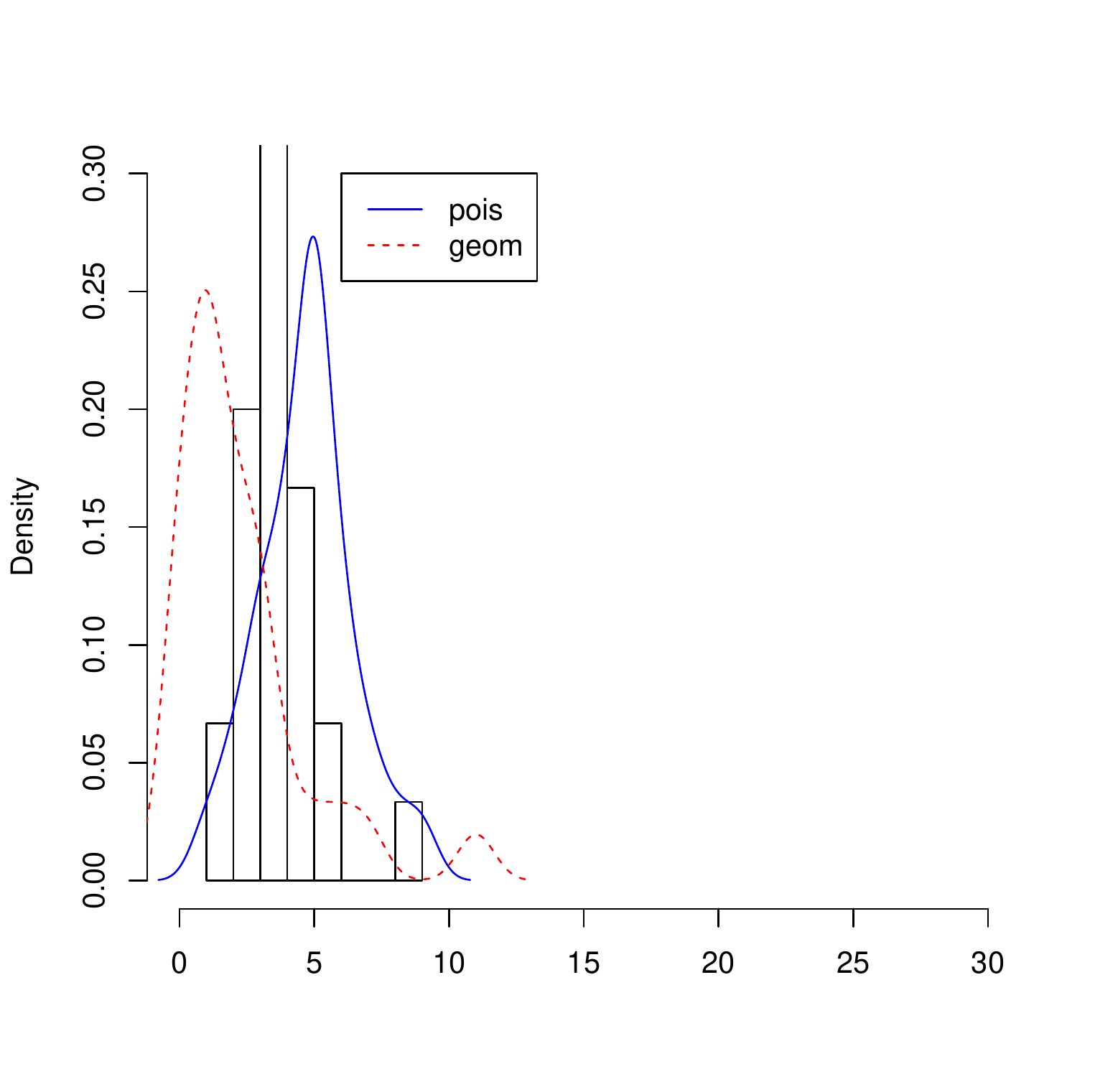}
}
\hspace{1cm}
\fbox{
\includegraphics[scale=0.34]{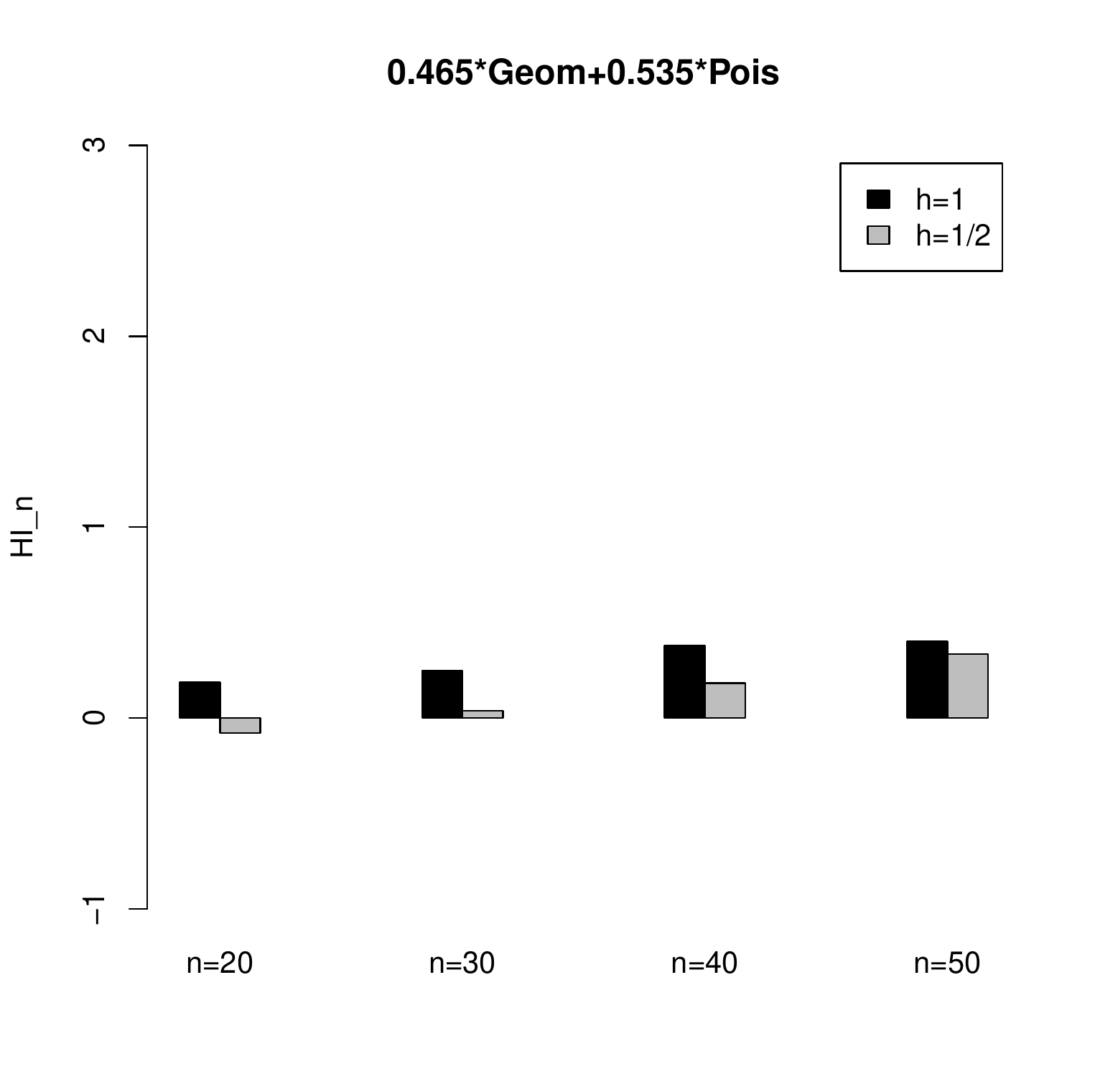}
}
\end{center}
\vspace{0.5cm}
\scriptsize {Figure 9 : Histogram of \\ DGP=0.465$\times$Geom+0.535$\times$Pois with n=50 } \hspace{0.5cm} {\scriptsize{Figure 10 : Comparative barplot of \,$HI_n$\, depending \,$n$ \\ }}
\label{hist5}
\end{figure}

Similarly, in figure 3, as can be seen, the geometric distribution closely approximates the data sets. 
In figures 5 and 7 the two distributions are close but the geometric (figure 5) and the poisson distributions (figure 7) does appear to be much closer to the data sets. When $\pi=0.535$, the distributions  for both (figure 9) poisson distribution and geometric distribution are similar, while being slightly symmetrical about the axis that passes through the mode of data distribution. This follows from the fact that these two distributions are equidistant from the DGP.

and would be difficult to distinguish from data in practice.  

 \begin{table}
 \begin{center}
\tiny{
\begin{tabular}{|p{2cm}|p{1cm}|p{0.5cm}|p{0.5cm}|p{0.5cm}|p{0.5cm}|p{0.5cm}|p{0.5cm}|p{0.5cm}|p{0.5cm}|p{0.5cm}|p{0.5cm}|}
\hline
\multicolumn{2}{|c|}{n}&\multicolumn{2}{|c|}{20}&\multicolumn{2}{|c|}{30}&\multicolumn{2}{|c|}{40}&\multicolumn{2}{|c|}{50}&\multicolumn{2}{|c|}{300}\\
\hline
\multicolumn{2}{|c|}{$\widehat{p}$}&\multicolumn{2}{|c|}{0.196(0.06)}&\multicolumn{2}{|c|}{0.204(0.05)}&\multicolumn{2}{|c|}{0.211(0.03)}&\multicolumn{2}{|c|}{0.213(0.207)}&\multicolumn{2}{|c|}{0.204(0.01)}\\
\hline
\multicolumn{2}{|c|}{$\widehat{\lambda}$}&\multicolumn{2}{|c|}{3.968(0.61)}&\multicolumn{2}{|c|}{3.962(0.46)}&\multicolumn{2}{|c|}{3.981(0.374)}&\multicolumn{2}{|c|}{4.023(0.309)}&\multicolumn{2}{|c|}{4.011(0.11)}\\
\hline \hline
DHP(Pois)&h=1&\multicolumn{2}{|c|}{2.869(0.63)}&\multicolumn{2}{|c|}{2.600(0.46)}&\multicolumn{2}{|c|}{2.582(0.36)}&\multicolumn{2}{|c|}{2.525(0.38)}&\multicolumn{2}{|c|}{2.311(0.25)}\\ \hline
\cline{2-10}
\multicolumn{1}{|c|}{}&h=1/2&\multicolumn{2}{|c|}{2.633(0.30)}&\multicolumn{2}{|c|}{2.492(0.28)}&\multicolumn{2}{|c|}{2.369(0.27)}&\multicolumn{2}{|c|}{2.302(0.26)}&\multicolumn{2}{|c|}{2.142(0.17)}\\
\hline
DHP(Geom)&h=1&\multicolumn{2}{|c|}{2.867(0.52)}&\multicolumn{2}{|c|}{2.682(0.37)}&\multicolumn{2}{|c|}{2.553(0.30)}&\multicolumn{2}{|c|}{2.495(0.20)}&\multicolumn{2}{|c|}{2.237(0.12)}\\ \hline
\cline{2-10}
\multicolumn{1}{|c|}{}&h=1/2&\multicolumn{2}{|c|}{2.157(0.21)}&\multicolumn{2}{|c|}{2.200(0.20)}&\multicolumn{2}{|c|}{2.263(0.20)}&\multicolumn{2}{|c|}{2.287(0.19)}&\multicolumn{2}{|c|}{2.291(0.15)}\\
\hline \hline
\multicolumn{1}{|c|}{$\mathcal{HI}^h$}&\multicolumn{1}{|c|}{$h=1/2$}&\multicolumn{2}{|c|}{-0.079(1.04)}&\multicolumn{2}{|c|}{0.038(1.05)}&\multicolumn{2}{|c|}{0.182(0.99)}&\multicolumn{2}{|c|}{0.334(1.10)}&\multicolumn{2}{|c|}{0.442(0.67)}\\ \hline
\cline{2-10}
\multicolumn{1}{|c|}{}&\multicolumn{1}{|c|}{Geom}&\multicolumn{2}{|c|}{03\%}&\multicolumn{2}{|c|}{04\%}&\multicolumn{2}{|c|}{05\%}&\multicolumn{2}{|c|}{10\%}&\multicolumn{2}{|c|}{13\%}\\
\multicolumn{1}{|c|}{}&\multicolumn{1}{|c|}{Indecisive}&\multicolumn{2}{|c|}{92\%}&\multicolumn{2}{|c|}{92\%}&\multicolumn{2}{|c|}{93\%}&\multicolumn{2}{|c|}{88\%}&\multicolumn{2}{|c|}{88\%}\\
\multicolumn{1}{|c|}{}&\multicolumn{1}{|c|}{Pois}&\multicolumn{2}{|c|}{05\%}&\multicolumn{2}{|c|}{04\%}&\multicolumn{2}{|c|}{02\%}&\multicolumn{2}{|c|}{02\%}&\multicolumn{2}{|c|}{01\%}\\ \hline
\hline
\multicolumn{1}{|c|}{$\mathcal{HI}^h$}&\multicolumn{1}{|c|}{$h=1$}&\multicolumn{2}{|c|}{0.186(1.14)}&\multicolumn{2}{|c|}{0.248(1.64)}&\multicolumn{2}{|c|}{0.378(0.90)}&\multicolumn{2}{|c|}{0.452(0.86)}&\multicolumn{2}{|c|}{0.617(0.73)}\\ \hline
\cline{2-10}
\multicolumn{1}{|c|}{}&\multicolumn{1}{|c|}{Geom}&\multicolumn{2}{|c|}{05\%}&\multicolumn{2}{|c|}{06\%}&\multicolumn{2}{|c|}{04\%}&\multicolumn{2}{|c|}{09\%}&\multicolumn{2}{|c|}{11\%}\\
\multicolumn{1}{|c|}{}&\multicolumn{1}{|c|}{Indecisive}&\multicolumn{2}{|c|}{92\%}&\multicolumn{2}{|c|}{90\%}&\multicolumn{2}{|c|}{95\%}&\multicolumn{2}{|c|}{90\%}&\multicolumn{2}{|c|}{88\%}\\
\multicolumn{1}{|c|}{}&\multicolumn{1}{|c|}{Pois}&\multicolumn{2}{|c|}{03\%}&\multicolumn{2}{|c|}{04\%}&\multicolumn{2}{|c|}{01\%}&\multicolumn{2}{|c|}{01\%}&\multicolumn{2}{|c|}{01\%}\\
\hline
\multicolumn{10}{c}{\small{\textit{\scriptsize {\ Table 5 : DGP=0.535$\times$Pois(4)+0.465$\times$Geom(0.2)}}}}
\end{tabular}}
\label{tab4}
\end{center}
\end{table}

The preceding results in tables and the theorem \eqref{th4} confirm, in figures 2, 4, 6 and 8, that 
the Hellinger indicator for the model selection procedure based on penalized hellinger divergence statistic with $h=0.5$ (light bars) dominates the procedure obtained with $h=1$ (dark bars) corresponding to the ordinary Hellinger distance. 
As expected, our statistic divergence  $\mathcal{HI}^h$ diverges to $-\infty$ (figure 2,  8) and to $+\infty$ (figure 4, figure 8) more rapidly when we use the penalized Hellinger distance test than the classical Hellinger distance test.\\Hence, 
Figure 10 \ allows a comparison with the asymptotic $\mathcal{N}(0,1)$ approximation under our null hypothesis of of equivalence. Hence the indicator $\mathcal{HI}^{1/2}$, based on the penaliezd Hellinger distance is closer to the mean of $\mathcal{N}(0,1)$ than is the indicator
$\mathcal{HI}^{1}$.

\section{Conclusion}
In this paper we investigated the problems of model selection using divergence type statistics. 
Specifically, we proposed some asymptotically standard normal and chi-square tests for model selection based on divergence type statistics that use the corresponding minimum penalized Hellinger estimator. Our tests are based on testing whether the competing models are equally close to the true distribution against the alternative hypotheses that one model is closer than the other where closeness of a model is measured according to the discrepancy implicit in the divergence type statistics used. The penalized Hellinger divergence criterion outperforms classical criteria for model selection based on the ordinary Hellinger distance, especially in small sample, the difference is expected to be minimal for large sample size. Our work can be extended in several directions. One extension is to use random instead of fixed cells. Random cells arise when the boundaries of each cell $c_i$ depend on some unknown parameter vector $\gamma$, which are estimated. For various examples, see e.g., Andrews (1988b). For instance, with appropriate random cells, the asymptotic distribution of a Pearson type statistic may become independent of the true parameter $\theta_0$ under correct specification. In view of this latter result, it is expected that our model selection test based on penalized Hellinger divergence measures will remain asymptotically normally or chi-square distributed. 

\bibliography{bsample}

\end{document}